\documentclass[xcolor,twocolumn,table,dvipsnames]{article}
\AtBeginDocument{%
  \providecommand\BibTeX{{%
    \normalfont B\kern-0.5em{\scshape i\kern-0.25em b}\kern-0.8em\TeX}}}

\usepackage{xcolor}
\usepackage{amsfonts}
\usepackage{amsthm}
\usepackage{amssymb}
\usepackage{url}
\usepackage{graphicx}
\usepackage{amsmath}
\usepackage{graphicx}
\usepackage{latexsym}
\usepackage{epstopdf}
\usepackage{multicol}
\usepackage{multirow}
\usepackage{fancyvrb}
\usepackage{subcaption}
\usepackage[skip=0.5pt]{caption}
\usepackage{tikz}
\usepackage{paralist}
\usepackage[vlined,ruled,commentsnumbered]{algorithm2e}
\let\oldnl\nl
\newcommand{\nonl}{\renewcommand{\nl}{\let\nl\oldnl}}%
\usepackage{tikz}
\usetikzlibrary{trees}
\usepackage{xcolor,listings}
\usepackage{textcomp}
\usepackage{forest}
\usepackage{mathtools}
\usepackage{enumitem}
\usepackage{tikz-qtree}
\usepackage{tabularx}
\usepackage{xcolor}
\usepackage{pifont}
\usepackage{standalone}
\usepackage{diagbox}

\usepackage[margin=0.7in]{geometry}

\newcommand{\daniel}[1]{\textcolor{red}{}}
\newcommand{\amir}[1]{\textcolor{orange}{}}
\newcommand{\ariel}[1]{\textcolor{blue}{}}
\newcommand{\yuval}[1]{\textcolor{purple}{}}

\newcommand{\reva}[1]{\leavevmode\color{black}{#1}}
\newcommand{\revb}[1]{\leavevmode\color{black}{#1}}
\newcommand{\revc}[1]{\leavevmode\color{black}{#1}}
\newcommand{\common}[1]{\leavevmode\color{black}{#1}}

\newcommand{\bbN}{\ensuremath{\mathbb{N}}}
\newcommand{\NX}{\ensuremath{\bbN[X]}}
\newcommand{\bbB}{\ensuremath{\mathbb{B}}}
\newcommand{\BX}{\ensuremath{\bbB[X]}}
\newcommand{\trio}{$Trio(X)$}
\newcommand{\posbool}{$PosBool(X)$}
\newcommand{\lin}{$Lin(X)$}
\newcommand{\why}{$Why(X)$}

\newcommand{\ex}{$K$-example}
\newcommand{\rel}{$K$-relation}
\newcommand{\db}{$K$-database}
\newcommand{\sM}{\!_M}
\newcommand{\sK}{\!_K}

\newcommand{\absEx}[1]{\ensuremath{\widetilde{#1}}}

\newtheorem{theorem}{Theorem}[section]
\newtheorem{proposition}[theorem]{Proposition}
\newtheorem{example}[theorem]{Example}

\newtheorem{definition}[theorem]{Definition}

\let\OLDthebibliography\thebibliography
\renewcommand\thebibliography[1]{
  \OLDthebibliography{#1}
  \setlength{\parskip}{0pt}
  \setlength{\itemsep}{0pt plus 0.3ex}
}

\definecolor{keywords}{rgb}{0,0,0.7}

\begin{document}


\title{On Optimizing the Trade-off between Privacy and Utility in \\ Data Provenance}





\author{
  Daniel Deutch \\ Tel Aviv University \\ \small danielde@tauex.tau.ac.il
  \and
  Ariel Frankenthal \\ Tel Aviv University \\ \small frankenthal@mail.tau.ac.il
  \and
  Amir Gilad\\ Duke University \\ \small agilad@cs.duke.edu
  \and
  Yuval Moskovitch \\ University of Michigan \\ \small yuvalm@umich.edu
}
\date{}

\maketitle

\begin{abstract}

Organizations that collect and analyze data may wish or be mandated by regulation to justify and explain their analysis results. At the same time, the {\em logic} that they have followed to analyze the data, i.e., their queries, may be proprietary and confidential. Data provenance, a record of the transformations that data underwent, was extensively studied as means of explanations. In contrast, only a few works have studied the tension between disclosing provenance and hiding the underlying query. 

This tension is the focus of the present paper, where we formalize and explore for the first time the tradeoff between the utility of presenting provenance information and the breach of privacy it poses with respect to the underlying query. Intuitively, our formalization is based on the notion of provenance abstraction, where the representation of some tuples in the provenance expressions is abstracted in a way that makes multiple tuples indistinguishable. The privacy of a chosen abstraction is then measured
based on how many queries match the obfuscated provenance, in the same vein as $k$-anonymity. The utility is measured based on the entropy of the abstraction, intuitively how much information is lost with respect to the actual tuples participating in the provenance. Our formalization yields a novel optimization problem of choosing the best abstraction in terms of this tradeoff. 
We show that the problem is intractable in general, but design greedy heuristics that exploit the provenance structure towards a practically efficient exploration of the search space. We experimentally prove the effectiveness of our solution using the TPC-H benchmark \reva{and the IMDB dataset}.



\end{abstract}

\section{Introduction}\label{sec:intro}


Data provenance, namely a record of the transformations that pieces of data underwent when processed by a query, has been the subject of extensive investigation in recent years \cite{trio,GKT-pods07,Userssemiring1,CheneyProvenance,ProvenanceBuneman,Olteanu12,Tan03}. Most of these works focus on the utility of provenance, showing that it is highly effective for applications such as hypothetical reasoning \cite{ArabGKRG16,AssadiKLT16,DeutchMR19}, explaining and justifying query results \cite{DeutchFG20,why,ChapmanJ09}, and others. The cost of provenance tracking is typically measured in terms of the execution time / memory overhead it incurs, and significant research effort has been dedicated to optimizing such computational aspects. 
In this paper, we shed light on a different kind of cost incurred by publishing provenance: the exposure of the {\em query} that has been executed and for which provenance has been tracked. We ask: {\em can we obfuscate provenance so that it remains useful, while hiding the underlying query?}  

\revc{
This aspect of provenance has become increasingly important as more and more agencies and organization aim to provide explanations for their decisions \cite{google-ads,facebook-ads} while governmental bodies and research communities stress the need for privacy-aware mechanisms \cite{privacy_opinion,un_privacy,ricciato2019trusted}. 
}



\begin{figure}[!htb]
    \centering \scriptsize
    \begin{minipage}{\linewidth}
        \centering
        \caption*{Interests}\label{tbl:trips}
        \begin{tabular}{ | c | c | c | c |}
            \cline{2-4}\multicolumn{1}{c|}{} & PID & Interest & Source\\
            \hline $i_1$ & 1 & Music & WikiLeaks\\
            \hline $i_2$ & 2 & Music & Facebook\\
            \hline $i_3$ & 3 & Music & LinkedIn\\
            \hline $i_4$ & 1 & Parties & WikiLeaks\\
            \hline $i_5$ & 2 & Parties & Facebook\\
            \hline $i_6$ & 4 & Movies & WikiLeaks\\
            \hline
        \end{tabular}
    \end{minipage}
    
    \begin{minipage}{\linewidth}
        \centering
        \caption*{Hobbies}\label{tbl:music}
        \begin{tabular}{ | c | c | c | c |}
            \cline{2-4}\multicolumn{1}{c|}{} & PID & Hobby & Source\\
            \hline $h_1$ & 1 & Dance & Facebook\\
            \hline $h_2$ & 2 & Dance & LinkedIn\\
            \hline $h_3$ & 4 & Dance & Facebook\\
            \hline $h_4$ & 1 & Trips & Facebook\\
            \hline $h_5$ & 2 & Trips & LinkedIn\\
            \hline $h_6$ & 3 & Trips & WikiLeaks\\
            \hline
        \end{tabular}
    \end{minipage}
    
    \begin{minipage}{1\linewidth}
        \centering
        \caption*{Persons}\label{tbl:persons}
        \begin{tabular}{ | c | c | c | c |}
            \cline{2-4}\multicolumn{1}{c|}{} & PID & Name & Age\\
            \hline $p_1$ & 1 & James T & 27\\
            \hline $p_2$ & 2 & Brenda P & 31\\
            \hline
        \end{tabular}
    \end{minipage}
    
    \caption{Partial Database instance of hobbies and interests of people collected from different sources}\label{fig:db}
\end{figure}

\begin{example}\label{ex:runningExample}
Consider an online advertising company that wishes to match ads to people. Their database contains information about people, their hobbies and interests, a sample of which appears in Figure~\ref{fig:db}. Each tuple has an identifier, appearing to its left. The company may run queries such as $Q_{real}$ appearing in Table \ref{fig:queries} looking for people that like dancing and music. The query output includes James and Brenda, and relevant advertisements may then be presented to them. Upon request, Brenda may receive an explanation of why the advertisement was shown to her (see e.g., \cite{google-ads,facebook-ads}). 
\revb{In the case where James and Brenda are friends, they may obtain each other explanation in addition to their own.}
However, the company may wish to avoid 
disclosing the general criteria (i.e., the query $Q_{real}$), since these criteria are part of the company's confidential business strategy.
\end{example}

The provenance of a given query result describes the tuples used by the query to derive the result and the manner in which they were used. 
We use here the well-established model of {\em provenance semirings} \cite{GKT-pods07}. 

\begin{table}
    \centering \footnotesize
    \caption{Queries for the running example. $Q_{real}$ is the original, $Q_{false1}$, $Q_{false2}$ are similar but not identical, and $Q_{general}$ is a generalization of the original}\label{fig:queries}
    \begin{tabularx}{\linewidth}{| c | X | c | c | c | c |}
        \hline {\bf Name} & {\bf Query} \\
        \hline $Q_{real}$ & $Q$(id) :- Person(id,name,age), Hobbies(id,`Dance',src1), Interests(id,`Music',src2) \\
        \hline $Q_{false1}$ & $Q$(id) :- Person(id,name,age), Hobbies(id,`Trips',src1), Interests(id,`Music',src2) \\
        \hline $Q_{false2}$ & $Q$(id) :- Person(id,name,age), Hobbies(id,`Dance',src1), Interests(id,`Parties',src2) \\
        \hline $Q_{general}$ & $Q$(id) :- Person(id,name,age), Hobbies(id,`Dance',src1), Interests(id,interest,src2) \\
        \hline
    \end{tabularx}
\end{table}

\begin{figure}[!htb]
    \centering \footnotesize
    \begin{subfigure}{\linewidth}
        \centering
        \begin{tabular}{| c | c |} 
            \hline Output & Provenance \\ [0.5ex]
            \hline 1 & $p_1 \cdot h_1 \cdot i_1$\\ 
            \hline 2 & $p_2 \cdot h_2 \cdot i_2$\\ 
            \hline
        \end{tabular}
        \subcaption{$Ex_{real}$} \label{fig:ex-real}
    \end{subfigure}
    
    \begin{subfigure}{\linewidth}
        \centering
        \begin{tabular}{| c | c |} 
            \hline Output & Provenance \\ [0.5ex]
            \hline 1 & $p_1 \cdot h_4 \cdot i_1$\\ 
            \hline 2 & $p_2 \cdot h_5 \cdot i_2$\\ 
            \hline
        \end{tabular}
        \subcaption{$Ex_{false1}$} \label{fig:ex-false1}
    \end{subfigure}
    
    \begin{subfigure}{\linewidth}
        \centering
        \begin{tabular}{| c | c |} 
            \hline Output & Provenance \\ [0.5ex]
            \hline 1 & $p_1 \cdot h_1 \cdot i_4$\\ 
            \hline 2 & $p_2 \cdot h_2 \cdot i_5$\\ 
            \hline
        \end{tabular}
        \subcaption{$Ex_{false2}$} \label{fig:ex-false2}
    \end{subfigure}
    \newline
    \caption{\ex s. $Ex_{real}$, $Ex_{false1}$ and $Ex_{false2}$ are the outputs of $Q_{real}$, $Q_{false1}$ and $Q_{false2}$, respectively} \label{fig:k-examples}
\end{figure}

\begin{example}\label{ex:prov_intro}
The provenance of the output tuple $(1)$ according to the query $Q_{real}$ shown in Table \ref{fig:queries} is presented in the first row of Figure \ref{fig:ex-real}. The expression, formulated as a product of the annotations $p_1, h_1, i_1$, intuitively means that the three tuples with these annotations in the database (Figure \ref{fig:db}) have jointly participated in an assignment to $Q_{real}$ that yielded this result.
\end{example}

We denote by {\em \ex} a subset (``example") of the results of a (hidden) query and an explanation for each result, formulated as its provenance (e.g., Figure \ref{fig:ex-real} shows \ex\ derived by $Q_{real}$, \revb{modeling the explanations for James and Brenda}).
Given a \ex, the problem we address is {\em how to modify the provenance in a way that still allows users to gain information from it, but without divulging the underlying query that produced it?}


We next detail the main components of our solution.






{\bf Obfuscating provenance through abstraction.~} We propose a simple way to obfuscate provenance, based on {\em provenance abstraction}. The main idea is to allow identification of multiple provenance annotations, replacing them with a common ``meta-annotation". Not all such identifications make sense in general, and so their choice is constrained by a tree whose leaves correspond to actual annotations and ancestors can be used as abstractions of their descendants. This technique has recently been proposed in \cite{DeutchMR19}, where it was used in a different context of reducing the provenance size.  


{\bf Quantifying loss of information.~} 
We use entropy \cite{shannon1948mathematical} to quantify the loss of information incurred by a choice of provenance abstraction. Information entropy expresses the level of uncertainty of a given data. In our context, we wish to measure ``how uncertain" is a viewer of the abstracted provenance expression, with respect to the actual one (each possibility for the actual provenance, given an abstraction, is called a concretization). We assume a given distribution over the concretizations. Lacking additional knowledge, this distribution may simply be taken as uniform. The entropy for an abstraction is then defined with respect to a tree and a distribution.      


{\bf Model for provenance privacy.~} Recall that our goal is to show an abstraction of a given \ex, while hiding the query that yielded the \ex.  To measure the privacy of an abstraction, we may thus look at the set of its possible concretizations, and then at the set of queries that would have yielded each concretization. In fact, not all such queries are ``interesting": we may restrict attention to 
connected inclusion-minimal queries \cite{DeutchG19}, i.e., queries whose join graph is connected and are not included in any other query in this set. These queries are representative of the viable options for the hidden query. We then define the privacy incurred by an abstraction as the cardinality of this set (i.e., how many connected inclusion-minimal queries match some concretization).


{\bf The problem of optimizing abstractions.~} 
The last two components are then combined to define the problem introduced and studied in this paper: given an example of query results and their provenance $Ex$, a provenance abstraction tree $T$, and a privacy threshold $k$, we aim at finding an abstraction that has at least $k$ connected inclusion minimal queries that `can fit' it, and minimizes the loss of information among all such abstractions. 



\begin{example}\label{ex:intro3}
Consider the \ex\ $Ex_{real}$ presented in Figure~\ref{fig:ex-real} showing two outputs of the query $Q_{real}$ and their provenance. The allowed abstractions are defined based on the tree $T$ depicted in Figure \ref{fig:abstree}. 
The leaves of $T$ are annotations (identifiers) of the tuples in Figure \ref{fig:db}, and its inner nodes are abstracted forms of these annotations. An abstraction of the provenance in $Ex_{real}$ w.r.t. $T$ may, e.g., replace the annotation $h_1$ with its ancestors {\em Facebook} or {\em Social Network}. Other tuple annotations may be abstracted as well. 
A choice of abstraction dictates a certain amount of information loss since the annotation {\em Facebook} can stand for any one of the annotations $h_1, h_3, h_4, i_2, i_5$, and when viewing the annotation {\em Facebook} we cannot be sure which annotation is the original. 
At the same time, it may obfuscate the underlying query $Q_{real}$, as more queries become consistent with the observable provenance information. 

\end{example}

\begin{figure}[ht]
\centering
\begin{tikzpicture}[thick, scale=.7,
    ssnode/.style={fill=mygreen,node distance=.1cm,draw,circle, minimum size=.1mm},
    -,shorten >= 2pt,shorten <= 2pt,
    level 1/.style={level distance=0.2cm},
    every  tree node/.style = {shape=rectangle, rounded corners, draw, align=center, top color=white, bottom color=blue!20},
    map/.style={shape=rectangle, rounded corners, draw, align=center, top color=white, bottom color=red!20},
    log/.style={shape=rectangle, rounded corners, draw, align=center, top color=white, bottom color=orange!60},
    level 2/.style={sibling distance=7mm},
    level 3/.style={sibling distance=1mm},
    level distance=1.1cm,
    level 4/.style={sibling distance=1mm},
    level 1/.style={level distance=0.7cm}
    ]
  \Tree
    [.\node[] (r) {$\star$};
    [.\node[] (a5) {Social Network}; 
    [.\node[] (a1) {Facebook}; 
    [.\node[] (a2) {$h_1$}; ]
    [.\node[] (a2) {$h_3$}; ]
    [.\node[] (a2) {$h_4$}; ]
    [.\node[] (a2) {$i_2$}; ]
    [.\node[] (a2) {$i_5$}; ]
    ]
    [.\node[] (a1) {LinkedIn}; 
    [.\node[] (a2) {$h_2$}; ]
    [.\node[] (a2) {$h_5$}; ]
    [.\node[] (a2) {$i_3$}; ]
    ]
    ]
    [.\node[] (a2) {WikiLeaks};
    [.\node[] (a2) {$h_6$}; ]
    [.\node[] (a2) {$i_1$}; ]
    [.\node[] (a2) {$i_4$}; ]
    [.\node[] (a2) {$i_6$}; ]
    ]
    ]
\end{tikzpicture}
\caption{Abstraction tree containing \reva{a subset of} tuple annotations in the database in Figure \ref{fig:db} as leaves, and inner nodes that are abstractions of the leaves}\label{fig:abstree}
\end{figure}
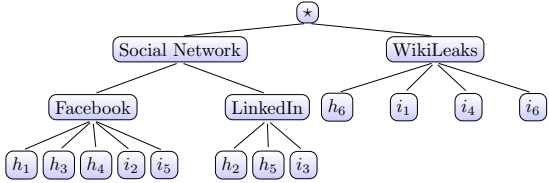

We study the complexity of the problem and show that it is intractable in general. Namely, deciding the existence of an abstraction with privacy at least $k$ and loss of information of at most $l$, is NP-hard.  Bearing this bound in mind, we provide novel heuristic algorithms for computing optimal abstractions in practically efficient ways. Our approach revolves around several key ideas. First, we optimize the order of traversal over the possible abstractions, by examining ``simpler" abstractions first. We further prioritize the computation of loss of information over privacy, as the former can be done significantly more efficiently. Additionally, privacy computation is performed in a greedy fashion, relying on the properties of the \ex. Finally, caching is used in order to avoid repetitive computations. Our heuristics and optimizations render our approach scalable even for large databases and complex queries, as observed in our experiments overviewed next.

{\bf Experimental evaluation.~} We have conducted an experimental study using the TPC-H \cite{tpch} \reva{and the IMDB \cite{imdb} datasets} in which we examined the scalability and usability of our solution for different settings. 
We study the performance in terms of varying data, tree sizes, query complexity, \ex\ size, and privacy thresholds. 
We show that thanks to our optimizations, our solution is efficient even in complex settings that involve queries with many joins, large volumes of data and a large space of abstractions.
\revb{We have also compared our solution with the provenance compression-based method presented in \cite{DeutchMR19}.} \common{Finally, we performed a user study, showing that abstracted \ex s provide the desired privacy while still being informative and useful.}



\section{Preliminaries}\label{sec:prelim}
We now define the background needed for our model. A summary of the notations used
throughout the paper is shown in Table \ref{notations}.

\begin{table}[!htb]
\centering \footnotesize

\caption{Notations}\label{notations}
\begin{center}
\begin{minipage}[t]{\linewidth}
\begin{center}
    \begin{tabular}{| c | l | c |}
    \hline $Q$ & \reva{Union} conjunctive query\\
    \hline $Ex$ & \ex \\
    \hline $T$ & Abstraction tree \\
    \hline $A_T$ & Abstraction function \\
    \hline $\absEx{Ex}$ & Abstracted \ex \\
    \hline $Var(Ex)$ & Set of variables in $Ex$ \\
    \hline $V_T$ & Set of nodes of tree $T$ \\
    \hline $L_T$ & Set of leaves of tree $T$ \\
    \hline $L_T(v)$ & \begin{tabular}{@{}c@{}}Set of leaves of the subtree\\of $T$ rooted in
    $v$\end{tabular} \\
    \hline $C(\absEx{Ex})$ & Concretization set of $\absEx{Ex}$ \\
    \hline
    \end{tabular}
    \end{center}
\end{minipage}
\end{center}
\end{table}


\subsection{Query Language and Provenance}\label{sec:cqs_prov}
We give a brief review of the concepts of Union of Conjunctive Queries and Provenance Polynomials.

{\bf Union of conjunctive queries.~}
We recall the concept of Unions of Conjunctive Queries. Fix a database schema $\mathcal{S}$ with relation names $\{R_1,...,R_n\}$ over a domain $\mathcal{C}$ of constants. Further fix a domain $\mathcal{V}$ of variables. A {\em CQ} $Q$ over $\mathcal{S}$ is an expression of the form $T(\vec{u}^{\,}) :- R_1(\vec{v_1}^{\,}), \ldots, R_l(\vec{v_l}^{\,})$ where $T$ is a relation name not in $\mathcal{S}$. For each $1 \leq i \leq n$, $\vec{v_i}^{\,}$ is a vector of the form $(x_1, \ldots, x_k)$ where $\forall 1 \leq j \leq k. \; x_j \in \mathcal{V} \cup \mathcal{C}$.
$T(\vec{u}^{\,})$ is the query head, denoted $head(Q)$, and $R_1(\vec{v_1}^{\,}), \ldots, R_l(\vec{v_l}^{\,})$ is the query body and is denoted $body(Q)$. The variables appearing in $\vec{u}^{\,}$ are called the \textit{head variables} of $Q$, and each of them must also appear in the body. A union of such queries is a UCQ. We use $UCQ$ to denote the class of all UCQs, omitting details of the schema when clear from the context.


Next, we define the notion of {\em derivations} for UCQs. A derivation
$\alpha$ for a query $Q \in UCQ$ with respect to a database instance
$D$ is a mapping of the relational atoms of $Q$ to tuples in $D$
that respects relation names and induces a mapping over arguments,
i.e., if a relational atom $R(x_1, ..., x_n)$ is mapped to a tuple
$R(a_1, ..., a_n)$ then we say that $x_i$ is mapped to $a_i$
(denoted $\alpha(x_i) = a_i$). We require that a variable $x_i$ will
not be mapped to multiple distinct values, and a constant $x_i$ will
be mapped to itself. For a CQ $q\in Q$, we define $\alpha(head(q))$ as the tuple
obtained from $head(q)$ by replacing each occurrence of a variable
$x_i$ by $\alpha(x_i)$.

\begin{example}\label{ex:cq}
Reconsider the CQ $Q_{real}$ depicted in Table \ref{fig:queries} and the output tuple (1) in the first row of Figure \ref{fig:ex-real}. It is derived using the tuples with annotations $p_1, h_1, i_1$ (Figure \ref{fig:db}) that are mapped to the first, second and third atom of $Q_{real}$ respectively. 
\end{example}



{\bf Provenance semirings.~} \reva{We focus on databases whose tuples are associated (``annotated") with elements of a set $X$, or polynomials (with positive coefficients) thereof \cite{GKT-pods07}. $X$ may be thought of as a set of identifiers each attached to a single input tuple. }

\reva{A {\em commutative monoid} (from \cite{DeutchG19}) is an algebraic structure $(M,+_{\sM},0_{\sM})$ where $+_{\sM}$ is an associative and
commutative binary operation and $0_{\sM}$ is an identity for $+_{\sM}$. A {\em commutative semiring} is
then a structure $(K,+_{\sK} ,\cdot_{\sK},0_{\sK},1_{\sK})$
where $(K,+_{\sK} ,0_{\sK})$ and $(K,\cdot_{\sK},1_{\sK})$ are commutative monoids, $\cdot_{\sK}$ is distributive over $+_{\sK} $,
and $a\cdot_{\sK}0_{\sK} = 0\cdot_{\sK} a = 0_{\sK}$.}
\reva{
A \rel\ is a mapping between tuples and elements of $K$.  A \db\ $D$ over a schema $\{R_1,...,R_n\}$ is then a collection of \rel s, over each~$R_i$. Unless stated otherwise, we will assume that in databases used as input to queries, all relations are abstractly-tagged: namely, each tuple is annotated by a distinct element of $X$ (intuitively, its identifier).
}

We then define \reva{UCQs} as mappings from \db s to \rel s. 
Intuitively, we define the annotation (provenance) of an output tuple as a combination of annotations of input tuples. 
The idea is that given a set of basic annotations $X$ (elements of which may be assigned to input tuples), the provenance of an output is represented by a sum of products, i.e., a polynomial. 
Coefficients serve in a sense as a ``shorthand" for multiple derivations using the same tuples, and exponents as a ``shorthand" for multiple uses of a tuple in a derivation.

\begin{definition} [adapted from \cite{GKT-pods07}]
\reva{
\label{def:basicprov} Let $D$ be a \db\ and let $Q \in UCQ$,
with $T_i$ being the relation name in $head(q_i)$ where $q_i \in Q$ is a CQ in $Q$. For every tuple $t \in T_i$, let $\alpha_{t}$ be the set of derivations of $q_i$ w.r.t. $D$ that yield $t$. $q_i(D)$ is defined to be a \rel\ $T_i$ s.t. for
every $t$, $T_i(t)=\sum_{head(q_i) = T_i}\sum_{\alpha \in \alpha_{t}}\prod_{t' \in
Im(\alpha)}Ann(t')$, where $Im(\alpha)$ is the image of $\alpha$, and $Ann(t')$ is the annotation of $t'$ according to its \rel. 
}
\end{definition}

\begin{example}
In Example \ref{ex:cq}, we showed that the output tuple (1) of $Q_{real}$ (Table \ref{fig:queries}) is derived from the tuples annotated by $p_1, h_1, i_1$. As a provenance polynomial, this corresponds to the monomial $p_1\cdot h_1\cdot i_1$. 
\end{example}

{\bf Provenance examples.~} We now define the notion of a \ex, which intuitively captures output examples and their explanations as provenance.




\begin{definition}[adapted from \cite{DeutchG19}]\label{provExample}
A \ex\ is a pair $(I,O)$ where $I$ is an abstractly-tagged \db\ called the {\em input} and $O$ is a \rel\ called the {\em output}.
\end{definition}

\reva{In words, $O$ denotes an output example and $I$ its provenance.}

\begin{example}\label{ex:ex}
A \ex\ is depicted in Figure \ref{fig:ex-real} where the left column shows two output examples, \reva{$O_1$ and $O_2$}, and the right column shows the provenance of each of them, \reva{$I_1$ and $I_2$}, respectively. 
\end{example}

For a \ex\ $Ex = (I,O)$, we denote by $Var(Ex)$ the set of tuple annotations in $I$ (see Table \ref{notations}).

\subsection{Provenance Abstraction Tree}
We define an {\em abstraction tree} over the provenance variables, drawing on \cite{DeutchMR19}. Intuitively, this defines groupings of different variables with a single value as a generalized representation of all of them.
The tree is structured so that the labels associated with tuples of the input examples are at the leaf level; inner nodes stand for abstractions of the labels associated with leaves of their sub-trees.

\begin{definition} \label{def:abs-tree}
An abstraction tree $T$ is a rooted labeled tree, where
each node has a unique label (we thus use ``node" and ``label" interchangeably). $V_T$ is used to denote the set of labels in $T$ and $L_T$ is the set of labels of the leaves in $T$. Given a \db\ $D$, we say that $T$ is {\em compatible} with $D$ if 
$(V_T\setminus L_T)\cap (\cup_{t\in D}Ann(t)) = \emptyset$.
\end{definition}



We say that an abstraction tree $T$ is {\em compatible} with a \ex\ $(I, O)$ if $T$ is compatible with $I$. 
\revb{If $T$ is not compatible with a \ex\ then it cannot be used as an abstraction tree for this particular \ex.}
\common{We will discuss ways of constructing abstraction trees at the end of Section \ref{sec:algo}.}

\begin{example}\label{ex:compatible}
Reconsider the \ex\ $Ex_{real}$ presented in Figure~\ref{fig:ex-real}. The abstraction tree $T$ shown in Figure~\ref{fig:abstree} is compatible with $Ex_{real}$ since none of the inner nodes of $T$ (e.g., $Facebook$) are labeled by the variables of $Ex_{real}$. 
\end{example}

\section{Model}\label{sec:model}
We define our novel model for the problem of provenance privacy.

\subsection{Abstractions and Concretizations}

Let $T$ be an abstraction tree. For $v,v' \in V_T$, we say that $v \leq_T v'$ if $v$ is a descendant of $v'$ in $T$ (or $v'=v$). 

\begin{definition} [Abstraction Function]
\label{def:absfunc}
Given an abstraction tree $T$ that is compatible with a \ex\ $Ex$ and an ordering over the variables of $Ex$ where each variable occurrence is assigned an index $i\in\bbN$, an {\em abstraction function} over $T$ is a function $A_T:Var(Ex) \times \bbN \rightarrow (V_T \cup Var(Ex))$ that maps each occurrence of a variable $v\in Var(Ex)$ at index $i$ such that $v\in L_T$ to $v'\in V_T$, where $v\leq_T v'$. If $v\notin L_T$, \revb{$A_T(v,i)=v$}.
\end{definition}

Note that $A_T$ may map different occurrences of the same variable $v$ to different nodes in $T$, namely, it is possible to have $A_T(v,i) \neq A_T(v,j)$, where $A_T(v,i)$ ($A_T(v,j)$) is the mapping of the $i$-th (resp. $j$) occurrence of $v$. To simplify notations, in the rest of the paper we assume each variable appears once, and omit the index from $A_T$.
Overloading notation, we use $A_T(Ex)$ to denote the \ex\ $\absEx{Ex}$ obtained by replacing each $v\in Var(Ex)$ by $A_T(v)$ for all $v\in L_T$. 

\reva{We next demonstrate the notion of abstraction function. In practice, these functions are generated automatically by the algorithm given in Section \ref{sec:algo}.}
In the rest of the paper, we will use the term {\em abstraction} interchangeably for the concepts of an abstraction function and its output, an abstracted \ex.

\begin{example}\label{ex:absfunc}
Reconsider the \ex\ $Ex_{real}$ given in Figure~\ref{fig:ex-real} and the abstraction function $A_T^1$ depicted in Figure \ref{fig:absfuncs}. Using $A_T^1$ on $Ex_{real}$ will create the abstracted \ex\ $\absEx{Ex}_{abs1}$ shown in Figure~\ref{fig:absexamples}. Formally, $A_T^1(Ex_{real}) = \absEx{Ex}_{abs1}$.
\end{example}


\begin{figure}[!htb]
\centering\scriptsize
\begin{minipage}{\linewidth}
	\centering
    \begin{subfigure}{\linewidth}
        \centering\scriptsize
		\[ A_T^1(v) = \begin{cases*}
            Facebook, & if $v = h_1, h_4$ \\
            LinkedIn, & if $v = h_2, h_5$ \\
            v,        & otherwise
        \end{cases*} \]
	\end{subfigure}
	\newline\newline
	\newline
    \begin{subfigure}{\linewidth}
        \centering\scriptsize
        \[ A_T^2(v) = \begin{cases*}
            WikiLeaks, & if $v = i_1, i_4$ \\
            Facebook,  & if $v = i_2, i_5$ \\
            v,         & otherwise
        \end{cases*} \]
	\end{subfigure}
	\newline\newline
	\newline
	\begin{subfigure}{\linewidth}
        \centering\scriptsize
        \[ A_T^3(v) = \begin{cases*}
            WikiLeaks, & if $v = i_1$ \\
            v,         & otherwise
        \end{cases*} \]
	\end{subfigure}
	\newline\newline\newline
	\caption{Abstraction Functions} \label{fig:absfuncs}
\end{minipage}

\begin{minipage}{\linewidth}
	\centering\scriptsize
    \begin{subfigure}{\linewidth}
        \centering\scriptsize
        \captionsetup{justification=raggedright,singlelinecheck=false}
        \caption*{\scriptsize{$\absEx{Ex}_{abs1}=A_T^1(Ex_{real})=A_T^1(Ex_{false1})=$}}
		\begin{tabular}{| c | c |} 
            \hline Output & Provenance \\ [0.5ex] 
            \hline 1 & $p_1 \cdot Facebook \cdot i_1$ \\
            \hline 2 & $p_2 \cdot LinkedIn \cdot i_2$ \\
            \hline
        \end{tabular} 
	\end{subfigure}
	\newline\newline
    \begin{subfigure}{\linewidth}
        \centering\scriptsize
        \captionsetup{justification=raggedright,singlelinecheck=false}
        \caption*{\scriptsize{$\absEx{Ex}_{abs2}=A_T^2(Ex_{real})=A_T^2(Ex_{false2})=$}} 
	    \begin{tabular}{| c | c |} 
            \hline Output & Provenance \\ [0.5ex] 
            \hline 1 & $p_1 \cdot h_1 \cdot WikiLeaks$ \\
            \hline 2 & $p_2 \cdot h_2 \cdot Facebook$ \\
            \hline
        \end{tabular}
	\end{subfigure}
	\newline\newline
	\begin{subfigure}{\linewidth}
        \scriptsize
        \captionsetup{justification=raggedright,singlelinecheck=false}
        \caption*{\scriptsize{$\absEx{Ex}_{abs3}=A_T^3(Ex_{real})=$}}
        \centering
		\begin{tabular}{| c | c |} 
            \hline Output & Provenance \\ [0.5ex] 
            \hline 1 & $p_1 \cdot h_1 \cdot WikiLeaks$ \\
            \hline 2 & $p_2 \cdot h_2 \cdot i_2$ \\
            \hline
        \end{tabular}
	\end{subfigure}
	\caption{Abstracted \ex s} \label{fig:absexamples}
\end{minipage}
\end{figure}

A concretization is then the `reverse' operation of abstraction. 

\begin{definition} [Concretization]
Given an abstracted \ex\ $\absEx{Ex}$ and an abstraction tree $T$, a \ex\ $Ex$ is a {\em concretization} of  $\absEx{Ex}$ if there exists an abstraction function $A_T$ such that $A_T(Ex)=\absEx{Ex}$. The concretization set of $\absEx{Ex}$ is
$C(\absEx{Ex})=\{Ex\mid\exists A_T.~ A_T(Ex)=\absEx{Ex}\}$
\end{definition}

Since sub-trees in the abstraction tree may have multiple leaves, an abstracted \ex\ can have more than one concretization. Therefore, we have defined the {\em concretization set} containing all options for concretizations.


\begin{example} \label{ex:concretization-set}
Consider again the abstracted \ex\ $\absEx{Ex}_{abs1}$  presented in Figure \ref{fig:absexamples}, the \ex\ $Ex_{real}$ shown in Figure \ref{fig:ex-real} and the abstraction function $A_T^1$ given in Figure \ref{fig:absfuncs}. 
From Example~\ref{ex:absfunc}, we have $Ex_{real} \in C(\absEx{Ex}_{abs1})$ since $A_T^1(Ex_{real})=\absEx{Ex}_{abs1}$.
Now consider the \ex\ $Ex_{false1}$ shown in Figure \ref{fig:ex-false1}. It also holds that  $A_T^1(Ex_{false1})=\absEx{Ex}_{abs1}$, and thus $Ex_{false1} \in C(\absEx{Ex}_{abs1})$, i.e., $Ex_{false1}$ is also in the concretization set of $\absEx{Ex}_{abs1}$. 
$C(\absEx{Ex}_{abs1})$ also contains other \ex s beside  $Ex_{real}$ and $Ex_{false1}$.
\end{example}

The following are simple observations regarding the size of a concretization set that will be useful in the sequel. Note that $L_T$ is the set of leaves of the abstraction tree $T$ and $L_T(v)$ is the set of leaves of the subtree of $T$ rooted in $v$.

\begin{proposition} \label{prop:concretization-set-size}
Given an abstraction tree $T$ that is compatible with a \ex\ $Ex$ and an abstraction function $A_T$, it holds that:
\begin{enumerate}
    \item $|C(A_T(Ex))|=\displaystyle\prod_{v \in Var(Ex)}~|~L_T(A_T(v))|$
    \item $1 \leq |C(A_T(Ex))| \leq |L_T|^n$, where $n=|\{v \in Var(Ex)~|~v \neq A_T(v)\}|$, and these bounds are tight.
\end{enumerate}
\end{proposition}

\begin{proof}
\begin{enumerate}[align=left, leftmargin=0pt, labelindent=\parindent, listparindent=\parindent, labelwidth=0pt, itemindent=!]
\item In induction on the number of abstracted values $n=|\{v \in Var(A_T(Ex))|v \ne A_T(v)\}|$. If $n=0$ it holds that $\forall v \in Var(A_T(Ex)), v=A_T(v)$. Thus, $\forall v$ it holds that
\begin{align*}
    \displaystyle\prod_{v \in Var(Ex)}|L_T(A_T(v))|&=\displaystyle\prod_{v \in Var(Ex)}|L_T(v)|\\
    &=\displaystyle\prod_{v \in Var(Ex)}|\{v\}| = 1
\end{align*}
It is also clear that $|C(A_T(Ex))|=1$ since the abstraction function is the identity function, so $Ex$ itself is the only concretization that holds $Id(Ex)=Ex$ and the base case is true.

About the inductive step, let's assume the proposition holds for $n$ and we will prove it for $n+1$. Let's denote $Var(Ex)=\{v_1,\dotsc,v_m\}, n<m$. Now, w.l.o.g, assume that if $i \in \{1,\dotsc,n\}$ then $v_i \ne A_T(v_i)$ and if $i \in \{n+1,\dotsc,m\}$ then $v_i = A_T(v_i)$. Now, by the inductive assumption it holds that:
\begin{align*}
    |C(A_T(Ex))|&=\displaystyle\prod_{v \in Var(Ex)}|L_T(A_T(v))|\\
    &=\displaystyle\prod_{i=1}^n|L_T(A_T(v_i))|
\end{align*}
The last equality holds since if $v=A_T(v)$ then
\begin{equation*}
    |L_T(A_T(v))|=|L_T(v)|=1
\end{equation*}
so it is not effect the product.

Now, for the $n+1$ case, we changed $A_T$ s.t.
\begin{equation*}
    v_{n+1} \in Var(A_T(Ex)), v_{n+1} \ne A_T(v_{n+1})
\end{equation*}
The concretization set contains only \ex s $Ex'$ that holds $\exists A_T, A_T(Ex')=Ex$. By definition, an abstraction function $A_T:L_T \rightarrow V_T$ is a function that transform each leaf $v$ to a single ancestor $v'$ in the tree. From that we know that if $Ex' \in C(Ex)$ it holds that
\begin{equation*}
    \forall v \in V(Ex'), v \ne A_T(v) \Rightarrow v \in L_T(v)
\end{equation*}
We also know that $v_{n+1} \ne A_T(v_{n+1})$ so it holds that $A_T(v_{n+1})$ can be any $l \in L_T(A_T(v_{n+1}))$, and since there are $|L_T(A_T(v_{n+1}))|$ options for that value, we multiple all the previous concretization with every new option. Thus,
\begin{align*}
    |C(A_T(Ex))|&=|L_T(A_T(v_{n+1}))| \cdot \displaystyle\prod_{i=1}^n|L_T(A_T(v_i))|\\
    &=\displaystyle\prod_{i=1}^{n+1}|L_T(A_T(v_i))|
\end{align*}
and we are done.

\item It is clear that $1 \leq |C(A_T(Ex))|$ since if the abstraction function is the identity function it is always true that $Id(Ex)=Ex$, so $Ex$ itself is a concretization.

Now, since $\forall v, L_T(A_T(v)) \leq L_T$, from the previous part we get that:
\begin{equation*}
    |C(A_T(Ex))|= \displaystyle\prod_{i=1}^n|L_T(A_T(v_i))| \leq \displaystyle\prod_{i=1}^n|L_T|=|L_T|^n
\end{equation*}
and we are done.

\item For the first equality, choosing $A_T'$ to be the identity function, i.e., $A_T'(v)=v$, the only concretization is $Ex$ itself, so $|C(A_T'(Ex))|=1$.

For the second equality, we denote by $r$ the abstraction tree $T$'s root. Consider the following abstraction function: $A_T''(v)=r, \forall v \in Var(Ex)$. With this abstraction tree, $|L_T(A_T''(v_i))|=|L_T|, \forall v \in Var(Ex)$, so it holds that:
\begin{equation*}
    |C(A_T''(Ex))|= \displaystyle\prod_{i=1}^n|L_T(A_T''(v_i))|= \displaystyle\prod_{i=1}^n|L_T|=|L_T|^n
\end{equation*}
\end{enumerate}
\noindent and we are done.
\end{proof}


\subsection{Loss of Information}\label{subsec:loss}

Each abstraction entails a loss of information. We 
measure the loss of information of an abstracted \ex\ $\absEx{Ex}$ via the notion of {\em Entropy}. Entropy is the average level of ``information'' or ``uncertainty'' inherent in the possible outcomes of a random variable \cite{shannon1948mathematical}. Given a random variable $X$, with possible outcomes $x_i$, each with probability $P_X(x_i)$, the entropy $H(X)$ of $X$ is as follows: $H(X)=-\sum_i P_X(x_i)\ln{P_X(x_i)}$. 
The entropy quantifies how `informative' or `surprising' the random variable is, averaged over all of its possible outcomes. Next, we define the entropy induced by abstraction, as follows:



\begin{definition}
Given an abstraction tree $T$ that is compatible with a \ex\ $Ex$, an abstraction function $A_T$ and a probability space on $X=C(A_T(Ex))$ (the concretization set of $A_T(Ex)$) we define the loss of information by
$LOI(A_T(Ex)) = -\sum_{i=1}^n P_{X}(x_i)\ln{P_{X}(x_i)}$
where $X=C(A_T(Ex))=\{x_1,\dotsc,x_n\}$ and $P_X(x_i)$ is the probability of the concretization $x_i$.
\end{definition}


The probabilities may be determined using statistical properties of the database or external information. Note that for a finite probability space $X$ with a discrete uniform distribution over $n$ states, the entropy is $H(X)=\ln(n)$. Since $C(A_T(Ex))$ is a finite set (Proposition \ref{prop:concretization-set-size}), if the probabilities of all concretizations in $C(A_T(Ex))$ are equal then $LOI(A_T(Ex)) = \ln(|C(A_T(Ex))|)$.

\begin{example}
Reconsider the abstracted \ex\ $Ex_{real}$ presented in Figure \ref{fig:ex-real}, the abstracted tree $T$ shown in Figure \ref{fig:abstree} and the abstraction function $A_T^3$ depicted in Figure \ref{fig:absfuncs}.
The output of $A_T^3(Ex_{real})$ is the abstracted \ex\ $\absEx{Ex}_{abs3}$ shown in Figure~\ref{fig:absexamples}.
The concretization set of $\absEx{Ex}_{abs3}$ is given in Figure \ref{fig:cons-set-ex-abs3}. Assuming the probabilities of the concretizations are $P_X(c_1)=0.1$, $P_X(c_2)=0.2$, $P_X(c_3)=0.3$ and $P_X(c_4)=0.4$. the loss of information of $\absEx{Ex}_{abs3}$ is 
$-\sum_{i=1}^4 P_{X}(c_i)\ln{P_{X}(c_i)} =-(0.1\cdot\ln0.1+\ldots+0.4\cdot\ln0.4) \approx1.279$
\end{example}




\subsection{Privacy}
We next define our privacy measure.

{\bf Consistent and CIM queries.~}
Next, we define the concepts of consistent and connected inclusion-minimal queries with respect to a \ex. Our definitions are inspired by \cite{DeutchG19} and extend them. As a preliminary step, we define subsumption of \rel s.


\begin{definition}[from \cite{DeutchG19}]\label{def:order}
\reva{Let $(K,+_{K},\cdot_{K},0,1)$ be a semiring and define $a \leq_{K} b$ iff  $\exists c.~ a+_{K} c = b$. If $\leq_{K}$ is a (partial) order relation then we say that $K$ is naturally ordered.
Given two $K$-relations $R_1,R_2$ we say that $R_1 \subseteq_{K} R_2$ iff 
$\forall t. R_1(t) \leq_{K} R_2(t)$.}
\end{definition}

We now define a consistent query w.r.t. an abstracted example. Intuitively, a query $Q$ is consistent w.r.t. $\absEx{Ex}$ if there exists a concretization of $\absEx{Ex}$ for which $Q$ generates the output tuples when given the provenance, and the provenance generated by $Q$ matches the one specified in the concretization. 

\begin{definition}\label{def:consistent}[consistent query]
Given an abstracted \ex\ $\absEx{Ex}$ and a CQ $Q$ we say that $Q$ is consistent with respect to the example $\absEx{Ex}$ if there exists $(I,O) \in C(\absEx{Ex})$ such that 
\reva{$O\subseteq_{K} Q(I)$}. 
\end{definition}

To define privacy, we use the concept of {\em connected inclusion-minimal queries} (CIM queries). 
Intuitively, we define the privacy criterion by the number of the most `focused' queries. We draw on previous works in the field of query-by-example \cite{KalashnikovLS18} that looks for connected queries and on \cite{DeutchG19} that looks for minimality in terms of inclusion. 
Recall that the join graph for a CQ is defined by the set of relations in its body $\{R_1, \ldots, R_m\}$ with an edge $(R_i, R_j)$ iff $R_i$ and $R_j$ share at least one variable. We say that a query is connected if its join graph is connected. 

\begin{definition}[CIM query]\label{def:cim_global}
A consistent query $Q$ with respect to a given abstracted \ex\ $\absEx{Ex}$ is a CIM query if it is connected and for every query $Q'$ such that $Q' \subsetneq_{K} Q$,
(i.e., for every \db\ $D$ it holds that $Q'(D) \subseteq_{K} Q(D)$, but not vice-versa), 
$Q'$ is not consistent with respect to $\absEx{Ex}$. Namely, $\forall Ex \in C(\absEx{Ex})$, $Q'$ is not a consistent query of $Ex$.
\end{definition}




\begin{example} \label{ex:cim-queries-all-conc}

Consider the abstracted \ex\ $\absEx{Ex}_{abs3}$ in Figure \ref{fig:absexamples} and its concretization set given in Figure \ref{fig:cons-set-ex-abs3}. There is only one CIM query w.r.t. $\absEx{Ex}_{abs3}$ which is $Q_{real}$ (shown in Table \ref{fig:queries}) since it is consistent w.r.t. the concretization $c_2$, connected and minimal w.r.t. all other consistent connected queries. 
Now consider the query $Q_{general}$ (shown in Table \ref{fig:queries}). It is consistent w.r.t. the concretization $c_3$ and connected. However, $Q_{general}$ is not CIM since $Q_{real} \subseteq Q_{general}$ (both queries have the same structure but $Q_{real}$ contains an extra constant).

\end{example}

\reva{Definition \ref{def:cim_global} may consider trivial queries as CIM if we allow for union. For example, in $\absEx{Ex}_{abs3}$ in Figure \ref{fig:absexamples}, the concretization $c_1$ in Figure \ref{fig:cons-set-ex-abs3} leads to the trivial CIM query $Q = q_1 \cup q_2$ where $q_1(1) :- p_1, h_1, h_6$ and $q_2(1) :- p_2, h_2, i_2$. Naturally, these types of UCQs do not generalize the \ex\ and therefore are not likely queries. In Section \ref{sec:algo}, we discuss a version of our solution that disqualifies such trivial queries.}

\begin{figure}[!htb]
    \centering \scriptsize
    \begin{multline*}
    C(\absEx{Ex}_{abs3})= \\ \left\{\begin{array}{lr}
        \substack{
        c_1 = \scriptsize{\begin{tabular}{| c | c |} 
            \hline Output & Provenance \\ [0.5ex]
            \hline 1 & $p_1 \cdot h_1 \cdot h_6$\\ 
            \hline 2 & $p_2 \cdot h_2 \cdot i_2$\\ 
            \hline
        \end{tabular}}\\\\\\
        c_2 = \scriptsize{\begin{tabular}{| c | c |} 
            \hline Output & Provenance \\ [0.5ex]
            \hline 1 & $p_1 \cdot h_1 \cdot i_1$\\ 
            \hline 2 & $p_2 \cdot h_2 \cdot i_2$\\ 
            \hline
        \end{tabular}}} &
        \substack{
        c_3 = \scriptsize{\begin{tabular}{| c | c |} 
            \hline Output & Provenance \\ [0.5ex]
            \hline 1 & $p_1 \cdot h_1 \cdot i_4$\\ 
            \hline 2 & $p_2 \cdot h_2 \cdot i_2$\\ 
            \hline
        \end{tabular}}\\\\\\
        c_4 = \scriptsize{\begin{tabular}{| c | c |} 
            \hline Output & Provenance \\ [0.5ex]
            \hline 1 & $p_1 \cdot h_1 \cdot i_6$\\ 
            \hline 2 & $p_2 \cdot h_2 \cdot i_2$\\ 
            \hline
        \end{tabular}}}
        \end{array}\right\}
    \end{multline*}
    \caption{\reva{Concretization Set of $\absEx{Ex}_{abs3}$ (from Figure \ref{fig:absexamples})}}\label{fig:cons-set-ex-abs3}
\end{figure}

{\bf Privacy of an abstracted \ex.~}
We are now ready to define the privacy of a \ex. Our definition is similar in spirit to the $k$-anonymity criterion in data privacy \cite{KAnonymity2002}. 

\begin{definition} [Privacy] \label{def:privacy}
The {\em privacy} of an abstracted \ex\ $\absEx{Ex}$ is the number of unique CIM queries w.r.t. $\absEx{Ex}$.
\end{definition}

\revb{As with $k$-anonymity, a higher number of unique CIM queries w.r.t. an abstracted \ex\ indicates that this abstracted \ex\ is more private.}
\reva{Even an abstracted \ex\ can reveal some information about the query structure. In particular, the tables participating in the query and possibly also the join structure can be inferred from the combination of the schema and the \ex.}


\begin{example} \label{ex:privacy}
Reconsider the abstracted \ex\ $\absEx{Ex}_{abs1}$ presented in Figure \ref{fig:absexamples}.  We now detail the CIM queries w.r.t. $\absEx{Ex}_{abs1}$. First, we note that the consistent queries w.r.t. $\absEx{Ex}_{abs1}$ are depicted in Table~\ref{fig:cons-queries-ex-abs1}. We choose only the queries that are connected (the queries marked by `con'). From these, we choose only the queries that are inclusion-minimal w.r.t. $\absEx{Ex}_{abs1}$. Those are the queries marked with `min' as well. Therefore, the CIM queries are annotated with `con, min'. There are only $2$ queries that fulfill these terms, $Q_{real}$ and $Q_{false1}$ (shown in Table~\ref{fig:queries}). Thus, the privacy of $\absEx{Ex}_{abs1}$ is $2$.
\end{example}

\begin{table}
    \centering \footnotesize
    \caption{Some of the consistent queries w.r.t. $\absEx{Ex}_{abs1}$ from Figure \ref{fig:absexamples}. There is a total of $14$ consistent queries. From those, $3$ are connected (labeled `con'), and from those $2$ are CIM (labeled `con, min'). This shows that the privacy of $\absEx{Ex}_{abs1}$ is $2$}\label{fig:cons-queries-ex-abs1}
    \begin{tabularx}{\linewidth}{| c | X | c | c | c | c |}
        \hline {\bf Class} & {\bf Query} \\
        \hline con, min & Q(a) :- Person(a,b,c), Hobbies(a,`Dance',d), Interests(a,`Music',e) \\
        \hline & Q(a) :- Person(a,p,q), Hobbies(r,s,t), Interests(u,v,w) \\
        \hline & Q(a) :- Person(a,b,c), Hobbies(d,`Dance',e), Interests(a,`Music',f) \\
        \hline con & Q(a) :- Person(a,b,c), Hobbies(a,d,e), Interests(a,`Music',f) \\
        \hline con, min & Q(a) :- Person(a,b,c), Hobbies(a,`Trips',d), Interests(a,`Music',e) \\
        \hline & Q(a) :- Person(a,b,c), Interests(d,`Music',e), Interests(a,`Music',f) \\
        \hline
    \end{tabularx}
\end{table}

Note that in Example \ref{ex:privacy}, all disconnected queries are missing the logic expressed by the connected queries. 

\subsection{Problem Definition}
We are now ready to define the problem of provenance abstraction. 
In short, given a \ex\ and a privacy threshold, we want to find an abstraction that satisfies this threshold but also minimizes the loss of information. 

\begin{definition}\label{def:problem} [Problem Definition]
Given an abstraction tree $T$ that is compatible with a \ex\ $Ex$ and $k \in \bbN$ a privacy threshold, our goal is to find an abstraction function $A_T$ where $A_T(Ex)$ has privacy $\geq k$, and $A_T$ minimizes $A_T(Ex)$'s loss of information out of all the abstraction functions that guarantee privacy $\geq k$. We call this abstraction an {\em optimal abstraction}.
\end{definition}

\begin{example}\label{ex:problem}
Reconsider the database depicted in Figure \ref{fig:db}, the query $Q_{real}$ shown in Table \ref{fig:queries}, its output $Ex_{real}$ given in Figure \ref{fig:ex-real} and the abstraction tree $T$ presented in Figure \ref{fig:abstree}. Assume that the privacy threshold is $2$ (i.e., we want our privacy to be at least $2$) and the loss of information is entropy with discrete uniform distribution.
We can use the abstraction function $A_T^2$ (detailed in Figure \ref{fig:absfuncs}) so that $A_T^2(Ex_{real})$ yields $\absEx{Ex}_{abs2}$ (depicted in Figure \ref{fig:absexamples}). Since the queries $Q_{real}$ and $Q_{false2}$ (shown in Table \ref{fig:queries}) are CIM w.r.t. $\absEx{Ex}_{abs2}$, its privacy is 2. 
In addition, $\ln|C(\absEx{Ex}_{abs2})|=\ln(5 \cdot 4)=\ln20\approx2.996$, thus the loss of information incurred by $A_T^2(Ex_{real})$ is $2.996$.
On the other hand, we can use the abstraction function $A_T^1$ (detailed in Figure \ref{fig:absfuncs}) so that $A_T^1(Ex_{real})$ yields $\absEx{Ex}_{abs1}$ (depicted in Figure \ref{fig:absexamples}). In Example \ref{ex:privacy} we have seen that the privacy of $\absEx{Ex}_{abs1}$ is $2$. In addition, $\ln|C(\absEx{Ex}_{abs1})|=\ln(5 \cdot 3)=\ln15\approx2.708$, thus the loss of information incurred by $A_T^1(Ex_{real})$ is $2.708$.
Since the loss of information of $A_T^1$ is smaller than all possible abstraction functions that guarantee privacy $\geq 2$ (in particular, $A_T^2$), it is an optimal abstraction.
\end{example}

\reva{
{\bf Aggregate queries.~} A model for provenance for aggregation queries was defined in~\cite{AmsterdamerDT11}. In a nutshell, the aggregation result is represented as a semimodule, that couples, using a tensor product, values from the aggregate domain and the tuple annotations. For example, consider an aggregate query with a similar structure to that of $Q_{real}$ (shown in Table \ref{fig:queries}), that performs a MAX aggregation on the age attribute, i.e., instead of the people ids it returns the maximal age of all people that like dancing and music. In this case the resulting aggregate value would be $(p_1\cdot h_1 \cdot i_1)\otimes 27 +_{MAX} (p_2\cdot h_2 \cdot i_2)\otimes 31$.
Our model can support queries with aggregation over the head variables, where abstraction functions operate on the tuple's annotation part in the semimodule. For instance, the result of applying $A_T^1$ (shown in Figure \ref{fig:absfuncs}) on the aforementioned aggregate result is $(p_1\cdot Facebook \cdot i_1)\otimes 27 +_{MAX} (p_2\cdot LinkedIn \cdot i_2)\otimes 31$.
}

\section{Hardness and Solution} \label{sec:algo}
We first note that the optimal abstraction problem is intractable. To this end we define the decision problem version of the optimal abstraction: given an abstraction tree compatible with a \ex\ and integers $k, l$, determine whether there is an abstraction function that gives a privacy of at least $k$ with at most $l$ loss of information. 
This decision problem is NP-hard in the size of the intersection of the provenance variables with the leaves of the abstraction tree. 

\begin{proposition}\label{prop:hardness}
The decision problem version of the optimal abstraction is NP-hard. 
\end{proposition}

\begin{figure}[!htb]
\begin{footnotesize}
\begin{subfigure}{0.5\linewidth}
\centering
\begin{tabular}{| c | c | c | c | c | c | c | c |}
\hline prov. & V & E & N\\
\hline $VC_1$ & $v_1$ & $e_i$ & 2 \\
\hline $\vdots$ & $\vdots$ & $\vdots$ & $\vdots$ \\
\hline $VC_{2m}$ & $v_n$ & $e_j$ & 2 \\
\hline yes & 0 & 1 & 3  \\
\hline
\end{tabular}
\caption{Relation $VC$}\label{fig:hardness_vc}
\end{subfigure}%
\begin{subfigure}{0.5\linewidth}
\centering
\begin{tabular}{| c | c | c | c | c | c | c | c |}
\hline prov. & J \\
\hline $E_1$ & $e_1$ \\
\hline $\vdots$ & $\vdots$ \\
\hline $E_{m}$ & $e_m$ \\
\hline ec & 1 \\
\hline
\end{tabular}
\caption{Relation $E$}\label{hardness_e}
\end{subfigure}
\end{footnotesize}
\caption{Database instance for the proof of Proposition \ref{prop:hardness}}
\end{figure}

\begin{figure}[ht]
    \centering \footnotesize
    \begin{tabular}{| c | c |} 
        \hline Output & Provenance \\ [0.5ex]
        \hline $e_1, \ldots, e_m$ & $VC_1 \cdots VC_{k} \cdot E_1 \cdots E_m$\\ 
        \hline $1, \ldots, 1$ & $yes^k \cdot ec^{m}$\\ 
        \hline
    \end{tabular}
\caption{\ex\ for the proof of Proposition \ref{prop:hardness} \amir{for a version without a consistent query w.r.t the original \ex, remove from the \ex\ the $E_1 \cdots E_m$ tuples from the first row and the $ec^m$ tuples from the second row}} \label{fig:hardness_example}
\end{figure}



\begin{proof}
We show that the problem is NP-hard by reduction from the decision problem version of Vertex Cover. The input to Vertex Cover is a graph $G = (V,E)$, where $|V| = n, |E| = m$ and an integer $k \in \mathbb{N}$, and the solution is a set of vertices $C \subseteq V$ such that $\forall e \in E.~ C\cap e \neq \emptyset$. 

Given such an input, we define the relation $VC(V,E,N)$, where $(v_i, e_j, 2) \in VC$ iff $v_i \in e_j$ (these tuples are denoted by $VC_1, \ldots, VC_{2m}$). $VC$ also contains the additional tuple $(0,1, 3)$, denoted by $yes$. The relation is shown in Figure \ref{fig:hardness_vc}.
Next, we define the relation $E$ depicted in Figure \ref{hardness_e}: for each edge $e_j\in E$, we have a tuple $E(e_j)$ denoted by $E_j$ and an additional tuple $E(1)$ denoted by $ec$. 

We then define a \ex\ $Ex$ with two rows as seen in Figure \ref{fig:hardness_example}, where $VC_1, \ldots, VC_k$ are chosen at random. Clearly, $Ex$ has a consistent query w.r.t. it which just projects the attributes of the atoms with relation $E$ to the output. 
We also define the abstraction tree to be $T$ where $L(T) = VC_1, \ldots, VC_N$ and each $VC_i$ is connected to a node $VC(\widetilde{v},\widetilde{e} , 2)$ (denoted by $\widetilde{VC}$), where the weight of each edge is $1$ \amir{or something with entropy}. 

Now, we claim that $G$ has a cover of size at most $k$ iff there is an  abstraction function that gives privacy at least $1$ with at most $k$ loss of information. 

($\Leftarrow$) Suppose we have an abstraction function $A_T$ that gives privacy at least $1$ with at most $k$ loss of information. Let $Q$ be a CIM query w.r.t. $A_T(Ex)$. In particular, $Q$ is consistent w.r.t. a certain concretization of $A_T(Ex)$ (Definition \ref{def:consistent}). Assume that the monomial in the first row of this concretization is $VC_1'\cdots VC_k'\cdot E_1 \cdots E_m$, like in this illustration:
\begin{center}
\begin{footnotesize}
\begin{tabular}{| c | c |} 
        \hline Output & Provenance \\ [0.5ex]
        \hline $e_1, \ldots, e_m$ & $VC_1'\cdots VC_k'\cdot E_1 \cdots E_m$\\ 
        \hline $1, \ldots, 1$ & $yes^k \cdot ec^{m}$\\ 
        \hline
\end{tabular}
\end{footnotesize}
\end{center}
Given this concretization, $Q$ should be connected, i.e., all atoms should have at least one join to another atom. Note that every atom with relation $E$ has to be connected to an atom with relation $VC$ (as they cannot be connected to each other). 
Suppose $Q$ is of the form:
\begin{equation*}\label{eq:shapConstraints}
\begin{split}
Q(x_1, \ldots, x_m) :- VC(y_1, x_1, z), \ldots, VC(y_k, x_m, z), \\E(x_1), \ldots, E(x_m)
\end{split}
\end{equation*}
This structure is necessary because each $E(x_j)$ has to be connected to some $VC(y_i, x_j, z)$ as it is the only option to create a connected query. 
Thus, given the provenance of the first row, $Q$ maps $x_1, \ldots, x_m$ to $e_1, \ldots, e_m$. 
We choose the vertices represented by $VC_1'\cdots VC_k'$ as the vertex cover for the graph $G$, since each tuple $VC_i' = VC(v_i,e_j, 2)$ represents a cover of $e_j$ by $v_i$ and all edges $e_1, \ldots, e_m$ appear in $VC_1'\cdots VC_k'$.

($\Rightarrow$) Suppose we have a cover of size $\leq k$, $\{v_1, \ldots, v_k\}$. 
We show how to generate an abstraction function from this cover. 
$A_T$ would abstract all $VC_1, \ldots, VC_k$ to $\widetilde{VC}$:
\begin{center}
\begin{footnotesize}
\begin{tabular}{| c | c |} 
        \hline Output & Provenance \\ [0.5ex]
        \hline $e_1, \ldots, e_m$ & $\widetilde{VC}^k\cdot E_1 \cdots E_m$\\ 
        \hline $1, \ldots, 1$ & $yes^k \cdot ec^{m}$\\ 
        \hline
\end{tabular}
\end{footnotesize}
\end{center}
Clearly, this gives $k$ loss of information. Next, we show that there is at least one consistent connected query, which, in particular, shows that there exists a CIM query (it may be contained in the query we show, but its existence is will be proved). 

First, we generate the concretization that our connected query will be consistent with. 
For every $v_i \in \{v_1, \ldots, v_k\}$, and the edge covered by it $e_j$, we replace an instance of $\widetilde{VC}$ with $VC(v_i, e_j, 2)$. This creates the concretization shown in the previous part of the proof. We now claim that the following query is connected and consistent w.r.t. this concretization:
\begin{equation*}\label{eq:shapConstraints}
\begin{split}
Q(x_1, \ldots, x_m) :- VC(y_1, x_1, z), \ldots, VC(y_k, x_m, z),\\ E(x_1), \ldots, E(x_m)
\end{split}
\end{equation*}
$Q$ is clearly connected. To see consistency, assign the $VC_i'$ tuples to the $VC$ atoms and the $E_j$ tuples to the $E$ atoms. 
\end{proof}

\reva{We have defined the problem for general semirings and UCQs (with aggregation). Now, we discuss the solution, starting from \NX\ and CQs. 
At the end of this section we consider other versions of the problem, where the provenance is given in a different model and the query class is more general}.
As shown above, the problem is intractable, and our algorithms incur exponential time in the worst case -- yet we design heuristics that significantly improve the performance in practice. We first give a high-level description of our solution and then introduce our algorithms.

\subsection{High Level Description} \label{sec:algo-high-level}
The brute force approach for solving the problem would go over all possible abstractions, compute the privacy and the loss of information of each and return the one with minimal loss of information among the ones that meet the privacy threshold. We next overview of how each of these components may be improved. 
The observed improvement over the brute force solution is reported in Section~\ref{sec:results}. 


{\bf Efficiently computing privacy.~}
The privacy computation is the most time consuming part of the solution (see Section \ref{sec:algo-details}). 
We next give an overview of how the privacy induced by a given abstraction may be efficiently computed. 
\begin{enumerate}[leftmargin=*,wide=0pt]
\item
{\em Computing privacy row by row.~}
Consistency with a \ex\ is monotone in the sense that each consistent query must be consistent with each subset of the rows in \ex. 
We use this fact to effectively compute privacy. For every abstracted \ex\ $\absEx{Ex}$, we first check whether the \ex\ containing only the first two rows of $\absEx{Ex}$ has at least $k$ CIM queries w.r.t. it, where $k$ is the privacy threshold. 
We store only concretizations of $\absEx{Ex}$ that admit consistent connected queries by storing which concretization creates each query. Then, we add the next row of $\absEx{Ex}$ to the stored concretizations from the previous step and repeat these steps. 
\item
{\em Concretizations connectivity.~}
We say that a \ex\ $Ex$ is connected if every provenance monomial in $Ex$ defines a connected graph where the nodes are the tuples and there is an edge between two tuples if they share a constant (e.g., $R(1,{\mathbf 2}), R({\mathbf 2},3)$ are connected). Observe that a connected consistent query cannot be obtained from a disconnected \ex; therefore, disconnected concretizations can be filtered out. 
\item
{\em Caching information about concretizations and queries.~}
Given two abstractions $\absEx{Ex}$, $\absEx{Ex}'$, it is common that $C(\absEx{Ex}) \cap C(\absEx{Ex}')$ contains multiple shared concretizations. Therefore, we use caching to store the consistent connected queries w.r.t. each concretization, to avoid repetitive computations (we do not store the CIM queries since the minimality of a query is measured w.r.t. the concretization set, which varies between different abstractions). Additionally, for each concretization, we store whether it is connected or not and use it in the following computations that involve this concretization.
\end{enumerate}

{\bf Efficiently finding an optimal abstraction.~}
Our next goal is to improve the n\"aive iteration over all abstractions. If we cleverly choose the {\em order in which we iterate over the abstractions} and avoid complicated calculations for irrelevant abstractions we can find a solution quickly. To do so, we use the following components.
In Section~\ref{sec:results}, we will show that these components have improved performance by a factor of over $500\times$.

\begin{enumerate}[leftmargin=*,wide=0pt]
    \item
{\em Sorting abstractions.~}
When we iterate over all the abstractions, we sort them in increasing order according to the number of tree edges they use, prioritizing abstractions with small loss of information. In this manner, abstractions that use fewer edges of the abstraction tree appear first (these are the easiest to compute privacy for since they have fewer concretizations).
Practically, such abstractions often meet the privacy threshold.
\item
{\em Prioritizing loss of information over privacy computation.~}
Unlike the loss of information that can be quickly and efficiently computed, computing the privacy of an abstracted \ex\ is a complex and pricey procedure (see Section \ref{sec:algo-details}). Therefore, given an abstracted \ex, we first compute the loss of information for each abstraction and only then compute the privacy. After finding the first abstraction that satisfies the privacy threshold, we only have to compute privacy for abstractions that incur less information loss.
\end{enumerate}


\subsection{Algorithm Details}\label{sec:algo-details}
We next detail the implementation of the ideas we have described.

{\bf Privacy computation.~}
We use the following components:
\begin{enumerate}[leftmargin=*,wide=0pt]
    \item\label{itm:consistent}{\em Finding consistent queries.~} To find all consistent queries w.r.t. a concretization we recall the algorithm $FindConsistentQuery$ from \cite{DeutchG19} that finds one consistent query for a given \ex\ by modeling the two provenance monomials of the first two rows in the \ex\ as a bipartite graph and finding partial matchings that `cover' the output attributes. The algorithm returns the first consistent query that is generated by such a matching. We adjust this algorithm to output all the consistent queries from all matchings instead of returning the first one we find. We then minimize each query using the lattice algorithm described in the paper. 
    
    \item {\em Finding minimal queries.~} 
    Given a set of queries $Q$, $q \in Q$ is minimal 
    if there is no query $q' \in Q$ such that $q' \subsetneq q$. We iterate over all the queries $q \in Q$, and for every $q' \in Q, q' \neq q$ we check whether $q' \subsetneq q$ using the procedure $QueryContainment$ that checks query containment (adapted from~\cite{CHEKURI2000211}).

\end{enumerate}

%
%

Algorithm \ref{algo:privacy-opt} computes the privacy of a given abstracted \ex\ $\absEx{Ex}$. The input is an abstracted \ex\ $\absEx{Ex}$ with $n$ rows, an abstraction tree $T$ and the privacy threshold $k$. The output is the privacy guaranteed by $\absEx{Ex}$, or $-1$ if the privacy is smaller than $k$. The algorithm initializes a set of good concretizations $GoodConc$ (concretizations that create consistent connected queries, as described in the `Computing privacy row by row' component in Section \ref{sec:algo-high-level}) with the first row of $\absEx{Ex}$ (line~\ref{ln:priv:good-conc}).
Then, it iterates over the rows in $\absEx{Ex}$ (lines \ref{ln:priv:rows-loop}--\ref{ln:priv:end-row-for}), and for each row preforms the following operations.
First, it collects the concretization sets of each abstracted \ex\ in $GoodConc$ combined with the current row from $\absEx{Ex}$ (lines \ref{ln:priv:conc-init}--\ref{ln:priv:get-conc}).
Second, it removes all the disconnected concretizations (line \ref{ln:priv:rem-conc-disc}) while for each concretization it uses caching to store whether it is connected or not, to avoid redundant computations.
Third, it collects all consistent queries w.r.t. every connected concretization and adds them to a set $Q_{cons}$ and to a map $QueriesToConc$ that stores, for each concretization, the queries that were created from it (lines \ref{ln:priv:init-cons-q-to-cons}--\ref{ln:priv:cons-add}).
Then, it removes all the disconnected queries from $Q_{cons}$ (line \ref{ln:priv:conn-q}) and also uses caching to store whether it is connected or not.
After that, it checks whether the number of connected queries is lower than our privacy threshold, and if so it returns $-1$ as the privacy does not satisfy the threshold (lines \ref{ln:priv:not-sat-priv-1}--\ref{ln:priv:not-sat-priv-1-end}).
Then, the algorithm re-sets the good concretization set $GoodConc$ with all the concretizations that create consistent connected queries using $QueriesToConc$ (lines \ref{ln:priv:good-conc-clear}--\ref{ln:priv:add-good-conc}). These concretizations will continue to the next iteration.
Finally, the algorithm selects only minimal queries (lines \ref{ln:priv:min-q}--\ref{ln:priv:end-row-for}) and checks again whether their number satisfies the privacy threshold (line \ref{ln:priv:not-sat-priv-2}). 
After the algorithm iterates over all rows, it returns the number of CIM queries (line \ref{ln:priv:output-priv}). 

\IncMargin{0.8em}
\begin{algorithm}[ht]
\footnotesize
    \caption{Compute Privacy}
    \label{algo:privacy-opt}
    \LinesNumbered
    
	\SetKwInOut{Input}{input}
	\Input{Abstracted $\ex\ \absEx{Ex}$, abstraction tree $T$, privacy threshold $k$}
	\SetKwInOut{Output}{output}
	\Output{The privacy of $\absEx{Ex}$ if it's at least $k$ or $-1$ otherwise}
	\BlankLine
	
	\nonl Let $\absEx{Ex}_i$ be the $i$th row of $\absEx{Ex}$ and $n$ be the number of rows of $\absEx{Ex}$\;
	$GoodConc \gets \{\absEx{Ex}_1\}$\; \label{ln:priv:good-conc}
	\For{$i \in \{2,\dotsc,n\}$}
	{\label{ln:priv:rows-loop}
	    $C \gets \emptyset$\;\label{ln:priv:conc-init}
	    \For{$gc \in GoodConc$}
	    {
	    \nonl $gc + \absEx{Ex}_i$ denotes appending the $i$'th row of $\absEx{Ex}$ to $gc$\;
	        $C \gets C \cup GetConcretizationSet(gc + \absEx{Ex}_i, T)$\;\label{ln:priv:get-conc}
	    }
	    $C_{connect} \gets RemoveDisconncted(C)$\; \label{ln:priv:rem-conc-disc}
        $Q_{cons} \gets \emptyset$;
        $QueriesToConc \gets (\emptyset, \emptyset)$\; \label{ln:priv:init-cons-q-to-cons}
        \For{$c \in C_{connect}$}
        {
            $Q_{cur} \gets GetConsistentsQueries(c)$\; \label{ln:priv:get-cons}
            $Q_{cons} \gets Q_{cons} \cup Q_{cur}$\;
            \For{$q \in Q_{cur}$}
            {
                $QueriesToConc \gets QueriesToConc \cup (q,c)$\; \label{ln:priv:add-q-to-cons}
            }
            \label{ln:priv:cons-add}
        }
        $Q_{conn} \gets GetConnectedQueries(Q_{cons})$\; \label{ln:priv:conn-q}
        \If{$|Q_{conn}| < k$}
        {\label{ln:priv:not-sat-priv-1}
            \Return $-1$\; \label{ln:priv:not-sat-priv-1-end}
        }
        $GoodConc \gets \emptyset$\; \label{ln:priv:good-conc-clear}
        \For{$q \in Q_{conn}$}
        {
            \For{$c \in QueriesToConc(q)$}
            {
                $GoodConc \gets GoodConc \cup \{c\}$\; \label{ln:priv:add-good-conc}
            }
        }
        $Q_{cim} \gets GetMinimalQueries(Q_{conn})$\; \label{ln:priv:min-q}
        \If{$|Q_{cim}| < k$}
        {\label{ln:priv:not-sat-priv-2}
            \Return $-1$\;
\label{ln:priv:end-row-for}        }
    }
    \Return $|Q_{cim}|$\; \label{ln:priv:output-priv}
\end{algorithm}
\DecMargin{0.8em}

\begin{example}\label{ex:privacy-algo}
Consider the \ex\ $Ex_{real}$, the tree $T$, the abstraction function $A_T^3$, and the abstracted \ex\ $\absEx{Ex}_{abs3}=A_T^3(Ex_{real})$ (depicted in Figures \ref{fig:ex-real}, \ref{fig:abstree}, \ref{fig:absfuncs}, and \ref{fig:absexamples}, respectively). Assume our privacy threshold is $2$ (i.e., we want our privacy to be at least 2).
First, the algorithm generates the concretization set $C(\absEx{Ex}_{abs3})$ (shown in Figure \ref{fig:cons-set-ex-abs3}) and removes the disconnected concretizations (which are \reva{$c_1$ and $c_4$}).
For each of the remaining concretization, the algorithm finds the consistent queries and amongst these, finds the CIM queries. As we saw in Example \ref{ex:cim-queries-all-conc}, after removing the disconnected queries we are left with $Q_{real}$ and $Q_{general}$ (shown in Table \ref{fig:queries}) and since $Q_{real} \subseteq Q_{general}$ there is only one CIM query which is $Q_{real}$ so the algorithm will return $-1$. 
\end{example}

{\bf Loss of information computation.~}
The loss of information can be easily computed given the abstracted \ex\ $\absEx{Ex}$ and the abstraction tree. If we use entropy with discrete uniform distribution, then the loss of information is equal to $\ln(|C(\absEx{Ex})|)$, i.e., the size of the concretization set. 
For other distributions, we can find the concretization set with the abstraction tree and calculate the entropy using the given distribution.


{\bf Optimal abstraction algorithm.~}
Given a \ex, an abstraction tree and a privacy threshold, Algorithm~\ref{algo:optimal-abs} finds the optimal abstraction which guarantees the threshold with minimal loss of information. First, the algorithm creates a set of all possible abstraction (line \ref{ln:opt:all-pos-abs}) and sorts it in increasing order by the number of edges in the abstraction tree used by each of the abstractions (ties are broken by their loss of information, line \ref{ln:opt:abs-sort}). Then it initializes the optimal abstraction to be $null$ and the optimal loss of information to be $\infty$ (line \ref{ln:opt:init}). For each abstraction, the algorithm first computes the loss of information (line \ref{ln:opt:loi}). If the loss of information is lower than the optimal loss observed, it computes the privacy (line \ref{ln:opt:privacy}), otherwise, it continues to the next abstraction. If the computed privacy meets the privacy threshold, the algorithm updates the optimal abstraction to be the current one, and updates the current optimal loss of information (lines \ref{ln:opt:new-best-start}--\ref{ln:opt:new-best-end}). Finally, it returns the abstraction that meets the privacy threshold and incurred the minimum loss of information (or $\emptyset$ if no abstraction has been found).

\IncMargin{1em}
\begin{algorithm}[ht]
\footnotesize
    \caption{Find Optimal Abstraction}
    \label{algo:optimal-abs}
    \LinesNumbered
    
	\SetKwInOut{Input}{input}
	\Input{\ex\ $Ex$, abstraction tree $T$, privacy threshold $k$}
	\SetKwInOut{Output}{output}
	\Output{Optimal abstraction}
	\BlankLine
	
    $A \gets AllPossibleAbstractions(Ex, T)$\;\label{ln:opt:all-pos-abs}
    $A_{sort} \gets SortAbstractions(A,T)$\;\label{ln:opt:abs-sort}
    $a_{best} \gets \varnothing$;
    $l_{best} \gets \infty$\; \label{ln:opt:init}
    \For{$a \in A_{sort}$}
    {
        $l \gets GetLossOfInformation(a)$\;\label{ln:opt:loi}
        \If{$l < l_{best}$}
        {\label{ln:opt:loi-condition}
            $p \gets ComputePrivacy(a)$\;\label{ln:opt:privacy}
            \If{$p \geq k$}
            {\label{ln:opt:new-best-start}
                $a_{best} \gets a$;
                $l_{best} \gets l$\; \label{ln:opt:new-best-end}
            }
        }
    }
    \Return $a_{best}$\;
\end{algorithm}
\DecMargin{1em}

\begin{example}
Reconsider the \ex\ $Ex_{real}$ and the abstraction tree $T$ (shown in Figures \ref{fig:ex-real} and \ref{fig:abstree} resp.). 
Assume that the privacy threshold is $2$ and the loss of information is entropy with discrete uniform distribution. 
First, the algorithm creates a set of all possible abstracted \ex s of $Ex_{real}$. Among these we have $\absEx{Ex}_{abs1}$ and $\absEx{Ex}_{abs3}$ (depicted in Figure \ref{fig:absexamples}). The corresponding abstraction functions are $A_T^1$ and $A_T^3$ (shown in Figure \ref{fig:absfuncs}).
The algorithm starts iterating all abstractions until it gets to $\absEx{Ex}_{abs3}=A_T^3(Ex_{real})$, which does not meet the threshold (as shown in Example \ref{ex:privacy-algo}). 
Then, the algorithm gets to $\absEx{Ex}_{abs1}=A_T^1(Ex_{real})$. Its privacy is $2$ (see Example \ref{ex:privacy}), satisfying the threshold. 
The loss of information is $\ln|C(\absEx{Ex}_{abs1})|=\ln(5 \cdot 3)=\ln15\approx2.708$ (Proposition \ref{prop:concretization-set-size}).
Since this is the first abstraction that meets the threshold, we keep it as the current optimal one.
The algorithm continues to iterate over all other abstractions for which the loss of information is smaller than the current optimal one. Since all of them do not satisfy the privacy threshold, it returns $\absEx{Ex}_{abs1}$ as the optimal abstraction.
\end{example}

{\bf Complexity.~}
Given a \ex\ $Ex$ and an abstraction tree $T$, the complexity of Algorithm \ref{algo:optimal-abs} for finding the optimal abstraction is $O((hl)^nq)$ 
where $h$ is the height of $T$, $l=|L_T|$ is the number of leaves in $T$, $n = |Var(Ex) \cap L_T|$ is the number of variables in $Ex$ that appears in $T$ and $q$ is an exponential expression in the arity (all considered queries have the same arity) which involves the consistency checks \cite{DeutchG19}, connectivity check and containment checks \cite{ChandraMerlin}.
First, the number of abstractions is $O(h^n)$ since there are $n$ variables that can be abstracted, and for each one of them we have $h$ options of abstracted values. Thus, for each abstraction we compute the concretization set which is of size $O(l^n)$ (since $|C(A_T(Ex))| \leq |L_T|^n$ from Proposition \ref{prop:concretization-set-size}). Finally, for each concretization we check for consistency, connectivity and containment in $O(q)$ where $q$ is exponential in the query arity. Our experimental evaluation that follows shows the practical efficiency of our solution.

\begin{table}[ht]
    \centering \footnotesize
    \caption{\reva{Privacy computation for the semirings (or semimodules) from \cite{Greenicdt09,AmsterdamerDT11} and different query classes. The approach we have detailed so far is designed for the scenario given in the gray cell and the modifications needed to adjust it to the other scenarios are given in the corresponding cells. The \lin\ semiring is discussed in the text}} \label{tbl:semirings_queries}
    \begin{tabularx}{\linewidth}{| c | X | X | c | c | c |}
        \cline{2-3}\multicolumn{1}{c|}{} & 
        {\bf \NX, \BX} & 
        {\bf \trio, \posbool, \why} \\
        \hline \multirow{2}{*}{CQ} & \multirow{2}{*}{\cellcolor{gray!25}Alg. \ref{algo:privacy-opt}} & \cellcolor{red!25} Change line \ref{ln:priv:get-cons} to Alg. 2 in \cite{DeutchG19} \\ 
        \cline{1-3} \multirow{2}{*}{UCQ, AGG} & \cellcolor{orange!25} Change lines \ref{ln:priv:conn-q} and \ref{ln:priv:min-q} &  
        \cellcolor{green!25} Change lines \ref{ln:priv:get-cons}, \ref{ln:priv:conn-q} and \ref{ln:priv:min-q} \\
        \hline
    \end{tabularx}
\end{table}

\reva{
{\bf Extending the solution.~}
Table \ref{tbl:semirings_queries} summarizes the augmentations needed for Algorithm \ref{algo:privacy-opt} when the provenance in the \ex\ is given in different semirings (table columns) and the query is permitted to be CQ, UCQ or aggregate query as specified in Section \ref{sec:cqs_prov} (table rows). 

{\em Gray cell.~}
First, for the \NX\ and \BX\ semirings, Algorithm \ref{algo:privacy-opt} does not need to be modified for CQs, as the \BX\ semiring simply drops coefficients from the polynomials and coefficients do not have an impact on the algorithm. 

{\em Orange cell.~}
For UCQs (and aggregate queries), line \ref{ln:priv:conn-q} needs to be adjusted to account for the definition of disconnected UCQ (a UCQ containing a disconnected CQ). Moreover, in line \ref{ln:priv:min-q} we may get CIM queries that are trivial, i.e., the simple union of the tuples that participate in the provenance of a concretization is a CIM query. Therefore, we can augment this procedure by eliminating such trivial queries by, e.g., changing Definition \ref{def:cim_global} that every CIM query has to have at least one variable. 

{\em Red cell.~}
The semirings \trio, \posbool and \why\ drop coefficients as well as powers and even monomials subsumed by other monomials (\posbool). The procedure for finding consistent queries in line \ref{ln:priv:get-cons}, therefore, needs to be adjusted to Algorithm 2 from \cite{DeutchG19} that finds consistent queries when given the provenance in these semirings. The algorithm accounts for the missing powers by expanding the provenance as much as needed until a consistent query is found. 
The algorithm proposed in \cite{DeutchG19} needs to be augmented as specified in Bullet (\ref{itm:consistent}) at the beginning of Section \ref{sec:algo-details}. 

{\em Green cell.~}
Similarly, for UCQs and aggregate queries, lines \ref{ln:priv:get-cons}, \ref{ln:priv:conn-q}, and \ref{ln:priv:min-q} have to change in the aforementioned manners. 

{\em The \lin\ semiring.~}
For the \lin\ semiring, adapting our solution is more challenging. 
This semiring incurs a significant loss of information about the query structure \cite{Greenicdt09}, both due to the nature of the semiring and due to the order relation in Definition~\ref{def:order}. For example, the provenance represented in the \NX\ semiring $2ab^2$ is represented as $\{a,b\}$. Furthermore, the order relation is translated to set containment, and thus, the provenance shown in the \ex\ can be any subset of the original set, i.e., the empty subset is also valid as provenance. }
\common{
If only part of the provenance set is given (i.e., there are missing tuples in the provenance set), we may employ an approach that `completes' the provenance in the most reasonable way for every concretization \cite{GiladM20} and then apply our solution as a subsequent step. If no provenance is given, we may be able to utilize methods from the field of query-by-example and query reverse-engineering \cite{Shen,KalashnikovLS18,joinQueries,qbo,TanZES17} to find the query structure strictly from the output, such as column mappings and candidate query generation. 
This will be the subject of future work.
}

\revb{
{\bf The dual problem.~}
The dual problem is defined as searching for the optimal abstraction whose loss of information does not exceed a certain threshold $l_{max}$. Algorithm \ref{algo:optimal-abs} can be adjusted to solve this problem using the following changes: (1) initializing $p_{best} \gets 0$ in line \ref{ln:opt:init} ($p_{best}$ will store the current optimal privacy), (2) changing the condition in line \ref{ln:opt:loi-condition} to be $l < min(l_{best}, l_{max})$ (this will limit the abstraction we scan to those which do not exceed the given threshold $l_{max}$), (3) changing the condition in line \ref{ln:opt:new-best-start} to be $p \geq p_{best}$ (this will optimize the privacy of the output abstraction) and (4) adding $p_{best} \gets p$ to line \ref{ln:opt:new-best-end} (this will update the current best privacy for the next abstractions we scan).
With those changes, the algorithm terminates if the loss of information exceeds $l_{max}$. This reduces the number of abstractions considered, thus the dual problem is more efficiently solvable.
}

\common{
{\bf Constructing abstraction trees.~}
Domain experts who know the database structure may be able to phrase rules that place annotations of similar tuples in proximity in the tree. For example, tuples containing the same values in the same attributes (e.g., Figure \ref{fig:abstree}), or are included in the same relation, etc. 
Another possible manner of constructing abstraction trees is based on {\em ontologies} that encode abstractions for the different tuples by grouping tuples with similar meaning. 
Existing methods for identifying semantic relationships between tuples may be used \cite{li2005learning,IdrissiBB15}. 
To further hone the constructed tree in terms of height and size, users could input the relevant queries and database to our system and try to adjust those parameters so that the system incurs the fastest runtime (see Figures \ref{fig:exp-tree-size-runtime} and \ref{fig:exp-tree-height-runtime} in Section \ref{section:experiments}).
The height can be adjusted, e.g., by adding or removing sub-categories in the ontology. The size can be modified by adding more tuples from the database to the tree. 
If the tree contains more tuple annotations, more abstractions are possible, which affects the possibility of finding an abstraction that meets the privacy threshold using less edges in the abstraction tree. 
}

\section{Experiments} \label{section:experiments}
We next detail the settings of our experimental study and its results. We further show end-to-end use cases of our framework. 





The algorithms were implemented in Java 13 using the TreeNode interface implementation to represent the abstraction trees.
All experiments were performed on Mac OS 10.15, 64-bit, with 16GB of RAM and Intel Quad-Core i7 2.2 GHz processor.

\subsection{Settings and Summary of the Results} \label{sec:exp-settings}
We next review the settings and the summary of our experiments.

{\bf Settings.~}
We study the scalability of our solution in terms of runtimes and the size of the optimal abstraction, i.e., the output of the algorithm (we measure the size as the number of edges in the abstraction tree that were used to get the optimal abstraction). 
For runtime experiments and the size of optimal abstraction experiments, we use the settings shown in Table \ref{tbl:scalability}.
To our knowledge, there is no comparable solution in previous work. We thus use the brute force approach as a baseline, studying the effects of each of our algorithm components described in Section \ref{sec:algo-high-level}. 
We have used the TPC-H dataset \cite{tpch} which consists of a suite of business oriented queries \reva{and the IMDB movies dataset \cite{imdb}}. We have randomly sampled a database of 1GB for all experiments.
Our basic settings is a privacy threshold of 5; 5-levels abstraction tree with 10000 leaves (10244 nodes); 2 rows in \ex; and discrete uniform distribution for the loss of information measure.

{\bf Abstraction trees.~} The TPC-H abstraction tree consists of a single relation `lineitem', randomly divided into subcategories evenly throughout the tree. \reva{The IMDB abstraction tree was created as follows:
(1) Directors and actors were categorized by their year of birth, which were further categorized by ranges of years. 
(2) Tables that connect actors and directors to movies were categorized similarly.
(3) Genres were categorized by the genre type.
(4) Movies were categorized by their released year, which were further categorized by ranges.
(5) Each one of the previous was categorized under a main category and all of those were categorized under the root.}

{\bf Queries.~} We have used the TPC-H queries 
whose details appear in Table~\ref{tbl:tpch-queries}. \reva{We have adapted those queries to our setting, i.e., we have converted them to CQs by dropping aggregation and arithmetics.} The queries are relatively complex (e.g., Q21 includes a triple self-join, i.e., a relation name occurring in 3 atoms). \reva{We also use the following IMDB queries:
(Q1) All the actors starring in a movie from 1995,
(Q2) All the actors who starred in a drama movie directed by an american director,
(Q3) All the actors which have a bacon number of 1 (actors who act in a movie with Kevin Bacon),
(Q4) All the directors which created an action movie and a comedy movie,
(Q5) All the comedy movies starred by an actor born in 1978,
(Q6) All the directors who directed a movie starring Tom Cruise,
and (Q7) All the actors who act in at least two action movies.}
All experiments were performed with all the queries. However, to avoid visual overloading in graphs and since the results of queries \reva{TPCH-Q5, TPCH-Q9, IMDB-Q3 and IMDB-Q4 were very similar to the results of queries TPCH-Q3, TPCH-Q7, IMDB-Q6 and IMDB-Q7} respectively, we omit their curves from the graphs.

{\bf Summary of the results.~}
\setdefaultleftmargin{0cm}{0cm}{}{}{}{}
\begin{compactenum}
\item Our solution scales well with the first three parameters in Table~\ref{tbl:scalability}, due to the components presented in Section \ref{sec:algo-high-level}. 
\item An increase in the number of rows in the \ex\ causes a significant runtime increase compared to the other parameters since Algorithm \ref{algo:optimal-abs} often has to iterate and analyze all possible abstractions, as in the brute force approach.
\item The tree height that yields minimum runtime for finding an optimal abstraction varies according to the query structure, though the number of required tree edges used 
steadily increases. 
\item As the size of the tree increases, the time for finding the optimal abstractions also increases, however, the number of required tree edges used for the abstraction decreases.  
\item Our solution is not sensitive to the loss of information distribution used, i.e., changing this parameter will not significantly change the runtime. However, the optimal abstraction may change since the distributions has changed, so another abstraction can now incur a smaller loss of information.
\item The effect of the components described in Section \ref{sec:algo-high-level} was dramatic in improving the scalability of our solution.
\item \revb{Compared to a provenance compression approach that also utilizes abstraction trees \cite{DeutchMR19}, our solution is able to output abstractions with a significantly lower loss of information.}
\item \common{We conducted a comprehensive user study, showing that users are unable to infer the original query from the abstracted \ex, while still being able to use the provenance to answer hypothetical questions about the data.}
\end{compactenum}

\begin{table}
    \centering \scriptsize
    \caption{Scalability experiments settings for Figures \ref{fig:exp-privacy-runtime}--\ref{fig:exp-ex-rows-runtime}}\label{tbl:scalability}
    \begin{tabular}{| c | c | c | c | c | c |}
        \hline {\bf Figures} &
        \begin{tabular}{@{}c@{}}{\bf Privacy}\\{\bf threshold}\end{tabular} & 
        \begin{tabular}{@{}c@{}}{\bf Abst.}\\{\bf tree size}\end{tabular} & 
        \begin{tabular}{@{}c@{}}{\bf Abst.}\\{\bf tree height}\end{tabular} & 
        \begin{tabular}{@{}c@{}}{\bf \# rows}\\{\bf in}\\{\bf \ex}\end{tabular} \\
        \hline \ref{fig:exp-privacy-runtime}, \ref{fig:exp-privacy-edges}, \ref{fig:exp-privacy-loi} & varying & 10244 & 5 & 2 \\
        \hline \ref{fig:exp-tree-size-runtime}, \ref{fig:exp-tree-size-edges} & 5 & varying & 5 & 2 \\
        \hline \ref{fig:exp-tree-height-runtime}, \ref{fig:exp-tree-height-edges} & 5 & 10244 & varying & 2 \\
        \hline \ref{fig:exp-joins-runtime} & 5 & 10244 & 5 & 2 \\
        \hline \ref{fig:exp-ex-rows-runtime} & 5 & 10244 & 5 & varying \\
        \hline
    \end{tabular}
\end{table}

\begin{table}
    \begin{center}
    \scriptsize
    \caption{TPC-H \reva{and IMDB} queries for the experiments}\label{tbl:tpch-queries}
    \begin{minipage}[t]{\linewidth}
    \centering
    \begin{tabular}{ | c | c | c | c |}
        \hline {\bf Query} & {\bf \# Atoms} & {\bf \# Joins} \\
        \hline TPCH-Q3 & 3 & 2 \\
        \hline TPCH-Q4 & 2 & 1 \\
        \hline TPCH-Q5 & 7 & 6 \\
        \hline TPCH-Q7 & 6 & 5 \\
        \hline TPCH-Q9 & 6 & 5 \\
        \hline TPCH-Q10 & 4 & 3 \\
        \hline TPCH-Q21 & 6 & 5 \\
        \hline
    \end{tabular}
    \end{minipage}\\
    \vspace{4mm}
    \begin{minipage}[t]{\linewidth}
    \centering
    \begin{tabular}{ | c | c | c | c |}
        \hline {\bf Query} & {\bf \# Atoms} & {\bf \# Joins} \\
        \hline IMDB-Q1 & 3 & 2 \\
        \hline IMDB-Q2 & 6 & 5 \\
        \hline IMDB-Q3 & 5 & 4 \\
        \hline IMDB-Q4 & 7 & 6 \\
        \hline IMDB-Q5 & 4 & 3 \\
        \hline IMDB-Q6 & 5 & 4 \\
        \hline IMDB-Q7 & 7 & 6 \\
        \hline
    \end{tabular}
    \end{minipage}
    \end{center} 
\end{table}

\subsection{Results} \label{sec:results}
We next detail our scalability results for the different settings.

{\bf Privacy threshold.~}
For this experiment we have increased the privacy threshold while fixing the other parameters (first row in Table \ref{tbl:scalability}). There are no strong and clear criteria on how to choose the privacy threshold exactly. For example, in the healthcare world when medical data is shared with $k$-anonymity property with a small number of people (typically for research purposes), $k$ is often chosen between 5 and 15. Thus, we have increased the privacy threshold from 2 to 20.
For privacy thresholds larger than 20, we noticed that the optimal abstraction returned had a significantly larger privacy than requested. For example, for a privacy threshold of 23, in 90\% of the runs the algorithm returned an optimal abstraction with at least $2\times$ privacy than requested (i.e., the number of CIM queries of the optimal abstraction was at least $2\times$ larger than the threshold). We have performed the following experiments:
\begin{enumerate}[leftmargin=*,label=(\alph*),wide=0pt]
    \item {\em Runtime.~} The results are shown in \reva{Figure \ref{fig:exp-privacy-runtime}} and indicate that our solution remains scalable even for a large privacy threshold.
    
    \item {\em Optimal abstraction size.~} We use `Optimal abstraction size' to represent the number of abstraction tree edges used in the optimal abstraction. The results are shown in \reva{Figure \ref{fig:exp-privacy-edges}} and indicate that we do not need a much larger abstraction to get larger privacy. We can see here that for TPCH-Q21 whose runtime was the slowest, we need fewer edges than for the other queries.
    
    \item {\em Loss of information.~} We study the loss of information as a function of varying privacy threshold. The results are shown in \reva{Figure~\ref{fig:exp-privacy-loi}} and indicate that the loss of information increases as privacy increases, as expected.
\end{enumerate}

\begin{figure}[!htb]
\vspace{-2mm}
	\centering
    \begin{subfigure}{0.49\linewidth}
        \centering
		\includegraphics[width = \linewidth]{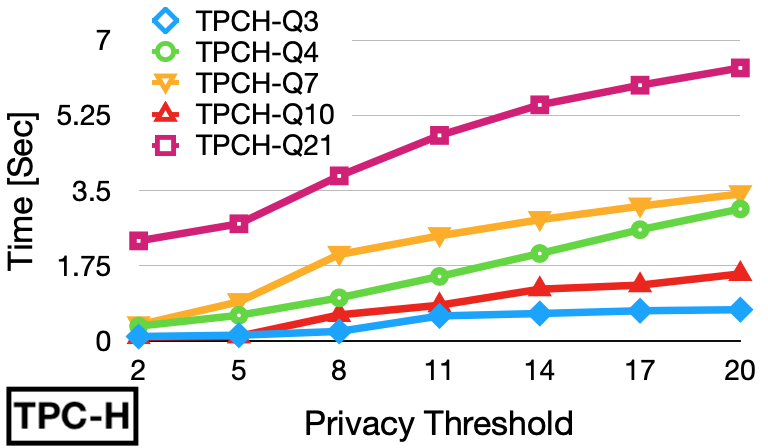}
	\end{subfigure}
	\hfill
	\begin{subfigure}{0.49\linewidth}
        \centering
		\includegraphics[width = \linewidth]{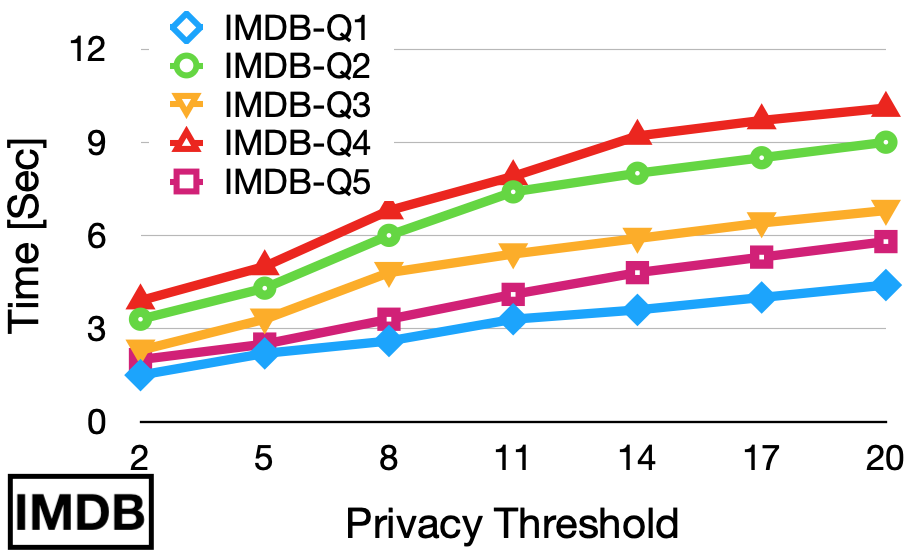}
	\end{subfigure}
	
	\caption{\reva{Runtime for varying number of privacy thresholds}}
	\label{fig:exp-privacy-runtime}
	\vspace{-6mm}
\end{figure}

\begin{figure}[!htb]
    \begin{subfigure}{0.49\linewidth}
        \centering
	    \includegraphics[width =    \linewidth]{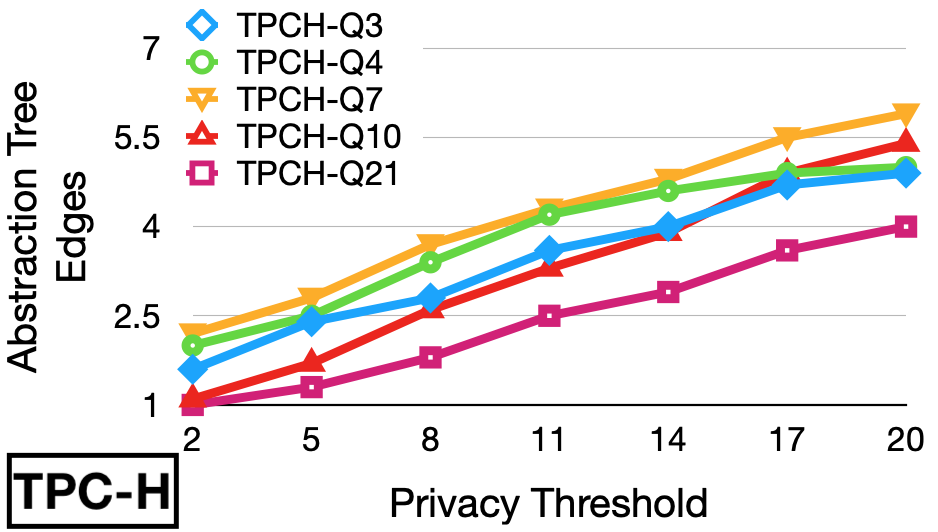}
	\end{subfigure}
	\begin{subfigure}{0.49\linewidth}
        \centering
	    \includegraphics[width =    \linewidth]{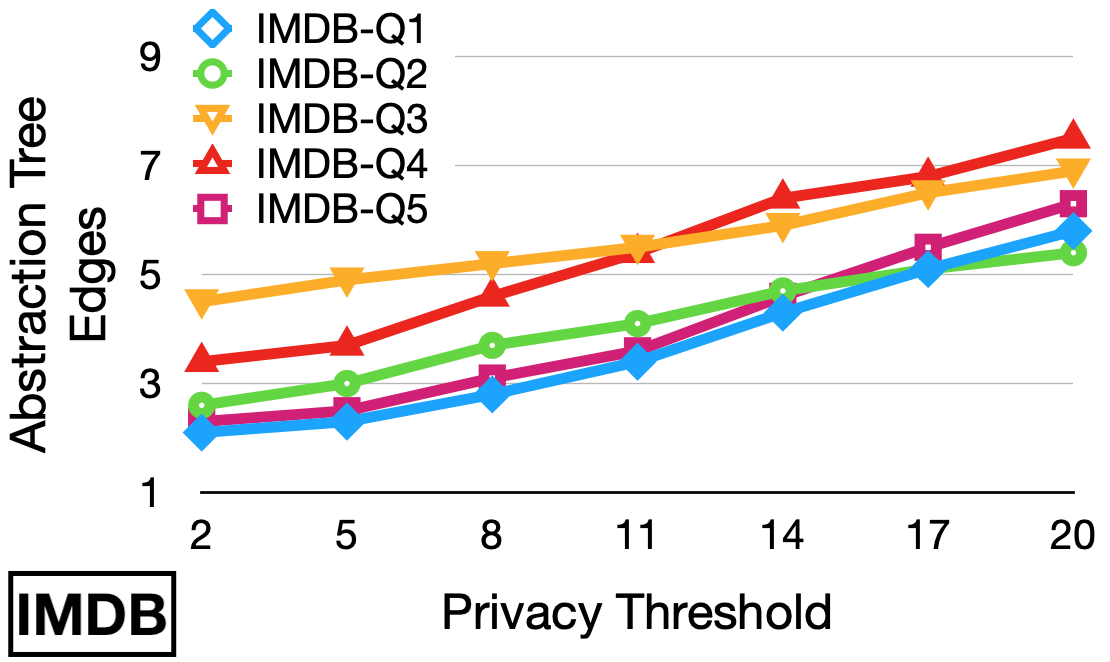}
	\end{subfigure}
	
	\caption{\reva{Optimal abstraction size for varying number of privacy thresholds}}
	\label{fig:exp-privacy-edges}
\end{figure}

\begin{figure}[!htb]
	\begin{subfigure}{0.49\linewidth}
        \centering
		\includegraphics[width = \linewidth]{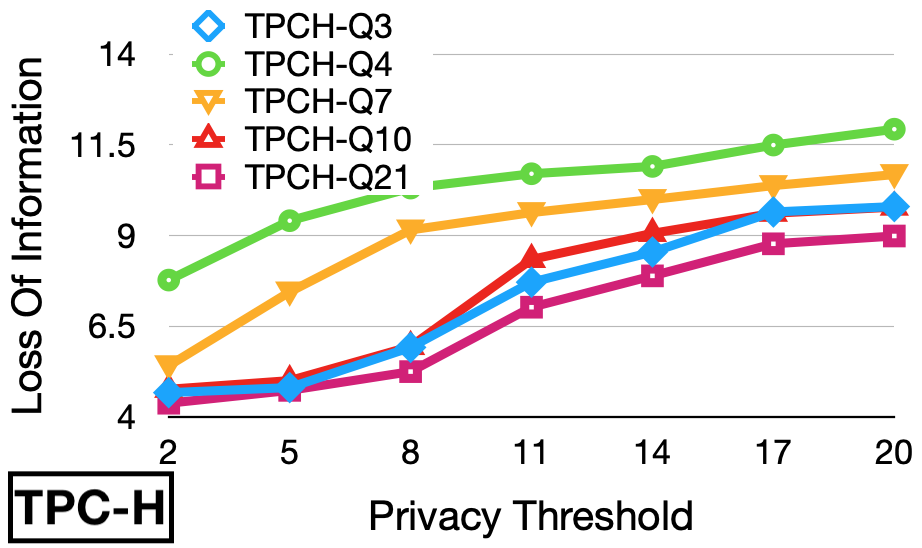}
	\end{subfigure}
	\begin{subfigure}{0.49\linewidth}
        \centering
		\includegraphics[width = \linewidth]{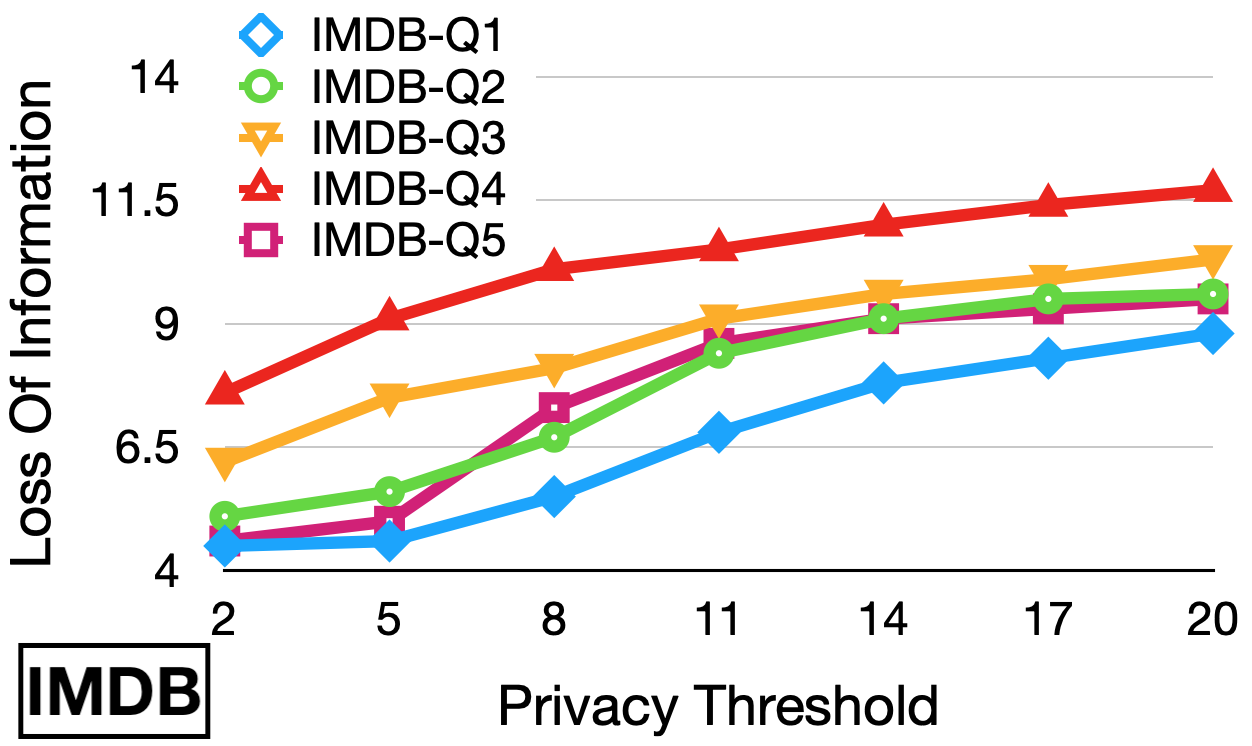}
	\end{subfigure}
	
	\caption{\reva{Loss of information for varying 
	privacy thresholds}}
	\label{fig:exp-privacy-loi}
	\vspace{-2mm}
\end{figure}

{\bf Abstraction tree size.~}
For this experiment we have increased the number of leaves in the tree from 10K to 810K. We have performed the following experiments:
\begin{enumerate}[leftmargin=*,label=(\alph*),wide=0pt]
    \item {\em Runtime.~} The results are shown in \reva{Figure \ref{fig:exp-tree-size-runtime}}. Our solution remains scalable even when the size of the abstraction tree nears the size of the data. We observed a similar trend when the tree size reached the data size. TPC-H queries Q3, Q5 and Q10 were faster than the rest since they have one `lineitem' atom which is connected to the rest of the query by a single attribute, as opposed to the other queries. Hence, there are fewer restrictions on these queries in terms of connectivity, making it easier to find CIM queries.
    
    \item {\em Optimal abstraction size.~} The results are shown in \reva{Figure \ref{fig:exp-tree-size-edges}} and indicate that when the abstraction tree is larger, the optimal abstraction requires fewer edges. The reason for this is that when the abstraction tree is larger there are more concretizations for each abstraction, and then the privacy can be larger for such abstractions. Here we have not directly measured Loss of Information since it depends on the tree structure which is varied here. 
\end{enumerate}

\begin{figure}[!htb]
	\centering
    \begin{subfigure}{0.49\linewidth}
        \centering
		\includegraphics[width = \linewidth]{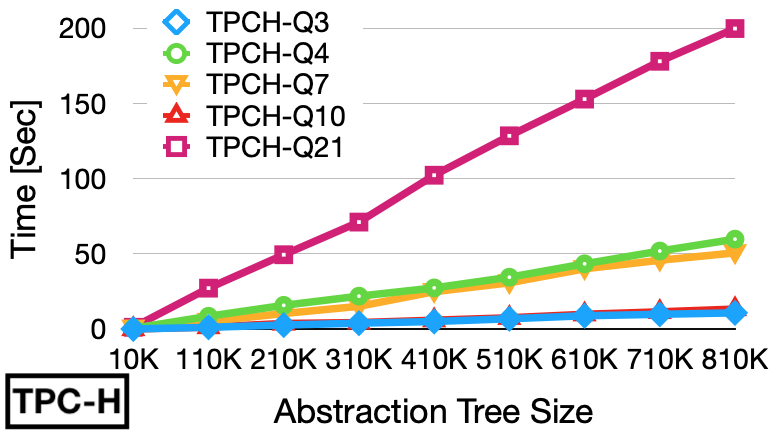}
	\end{subfigure}
	\hfill
	\begin{subfigure}{0.49\linewidth}
        \centering
		\includegraphics[width = \linewidth]{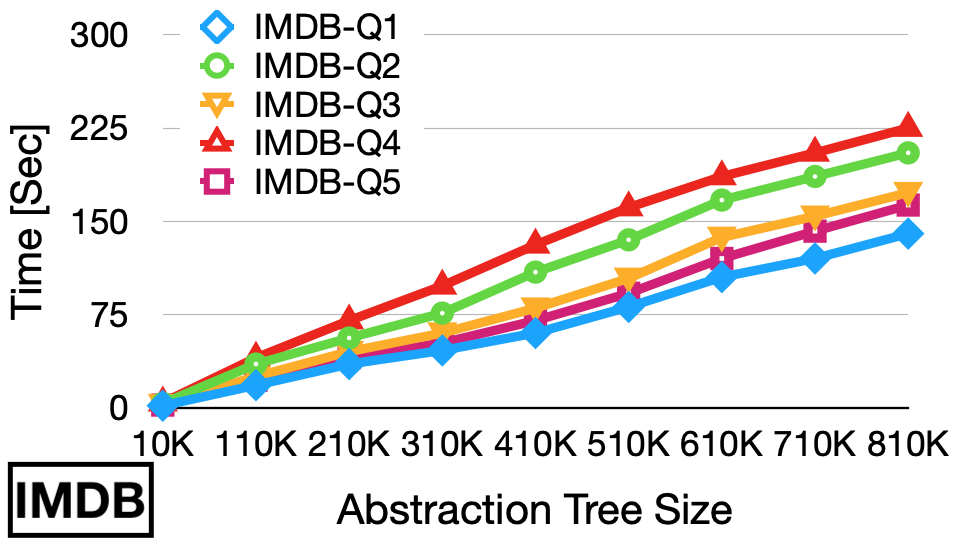}
	\end{subfigure}
    
	\caption{\reva{Runtime for varying abstraction tree size}}\label{fig:exp-tree-size-runtime}
	\vspace{-3mm}
\end{figure}

\begin{figure}[!htb]
	\centering
    \begin{subfigure}{0.49\linewidth}
        \centering
    	\includegraphics[width = \linewidth]{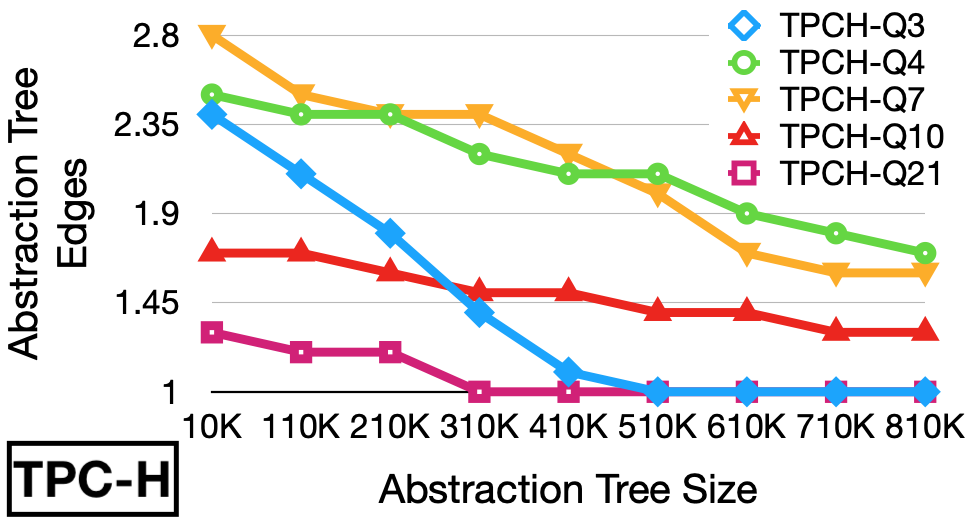}
    \end{subfigure}
	\hfill
	\begin{subfigure}{0.49\linewidth}
        \centering
    	\includegraphics[width = \linewidth]{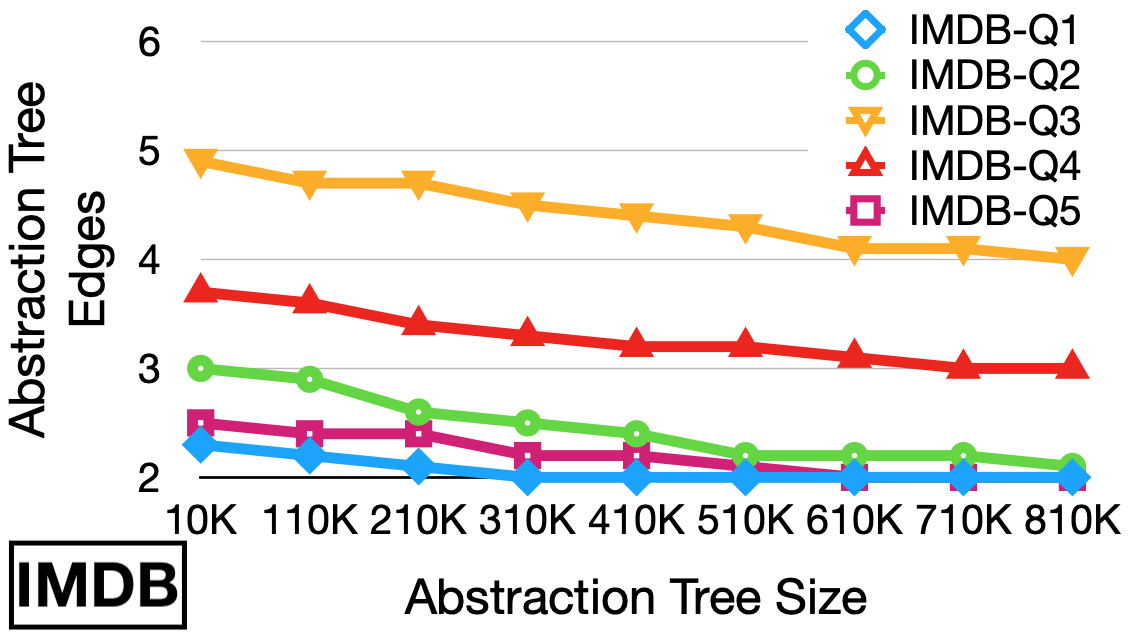}
    \end{subfigure}
    
	\caption{\reva{Optimal abstraction size for varying 
	tree size}}\label{fig:exp-tree-size-edges}
\end{figure}

{\bf Abstraction tree height.~}
We next examined the abstraction tree height. We have performed the following experiments:
\begin{enumerate}[leftmargin=*,label=(\alph*),wide=0pt]
    \item {\em Runtime.~} The results are shown in \reva{Figure \ref{fig:exp-tree-height-runtime}}. Interestingly, we noticed that every query has an optimal height for which the runtimes are the fastest (e.g., for TPCH-Q7, the optimal height is 5). Particularly, there is no trend of the sort ``higher tree implies longer runtime to find the optimal abstraction". Instead, the tree height that yields the fastest runtime is dependent on the query structure.
    \item {\em Optimal abstraction size.~} The results are shown in \reva{Figure \ref{fig:exp-tree-height-edges}} and indicate that the optimal abstraction size increases when the tree height increases.
\end{enumerate}
We have observed that different queries require traversing a different number of concretizations to achieve the desired privacy. If the query is relatively simple (e.g., TPCH-Q4) it needs less and if the query is relatively complicated (e.g., TPCH-Q21) it needs more. 
On the one hand, if the tree is not sufficiently high, every abstraction has more concretizations than we need, so the runtime will be slower. On the other hand, if the tree is too high, every abstraction has fewer concretizations than we need, so we have to scan more abstractions to find a solution and the runtime will also be slower.

\begin{figure}[!htb]
	\centering
	\begin{subfigure}{0.49\linewidth}
        \centering
	    \includegraphics[width =    \linewidth]{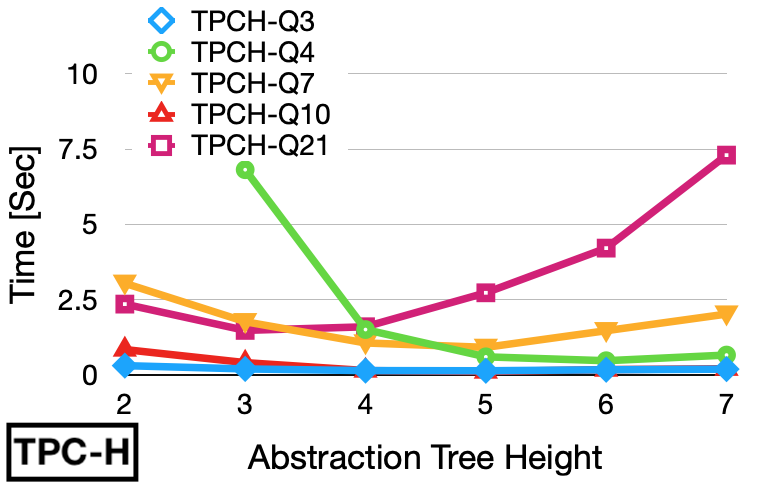}
	\end{subfigure}
	\hfill
	\begin{subfigure}{0.49\linewidth}
        \centering
	    \includegraphics[width =    \linewidth]{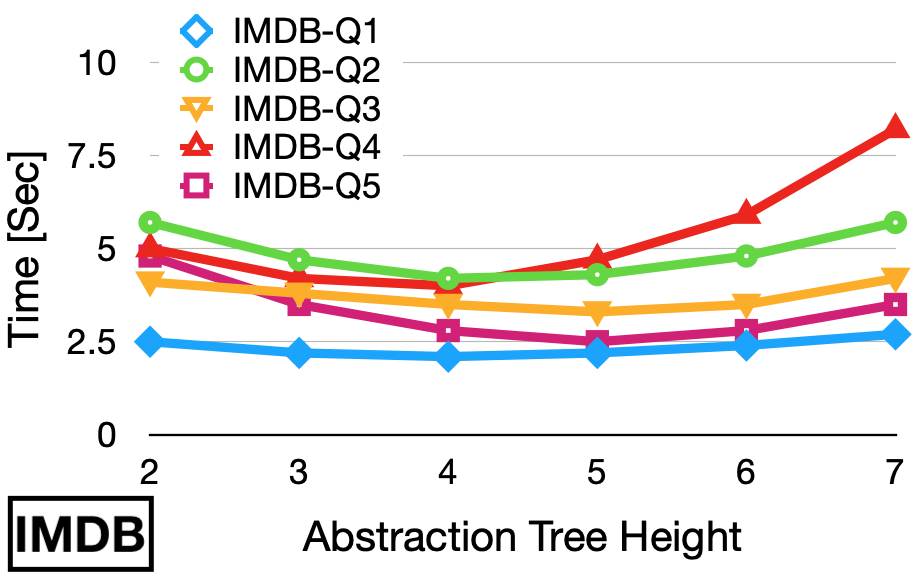}
	\end{subfigure}
	
	\caption{\reva{Runtime for varying abstraction tree height}}\label{fig:exp-tree-height-runtime}
	\vspace{-7mm}
\end{figure}

\begin{figure}[!htb]
	\centering
	\begin{subfigure}{0.49\linewidth}
        \centering
		\includegraphics[width = \linewidth]{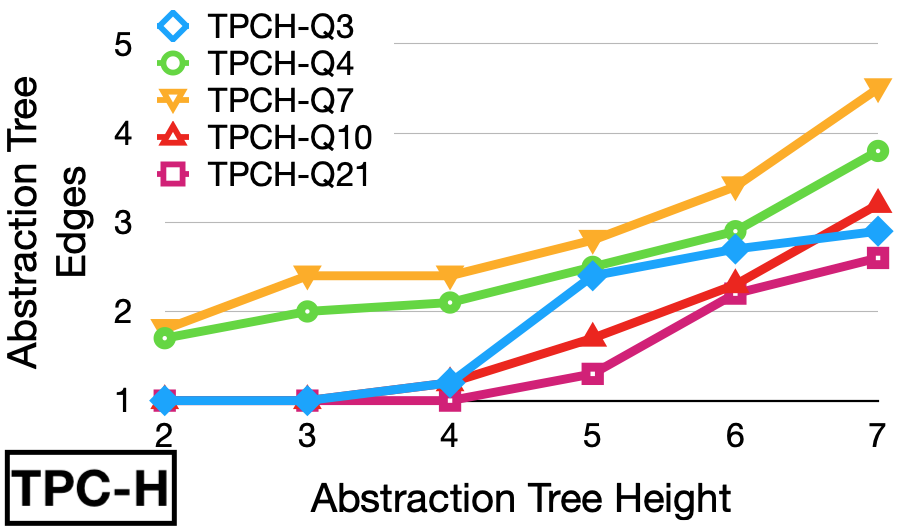}
	\end{subfigure}
	\hfill
	\begin{subfigure}{0.49\linewidth}
        \centering
		\includegraphics[width = \linewidth]{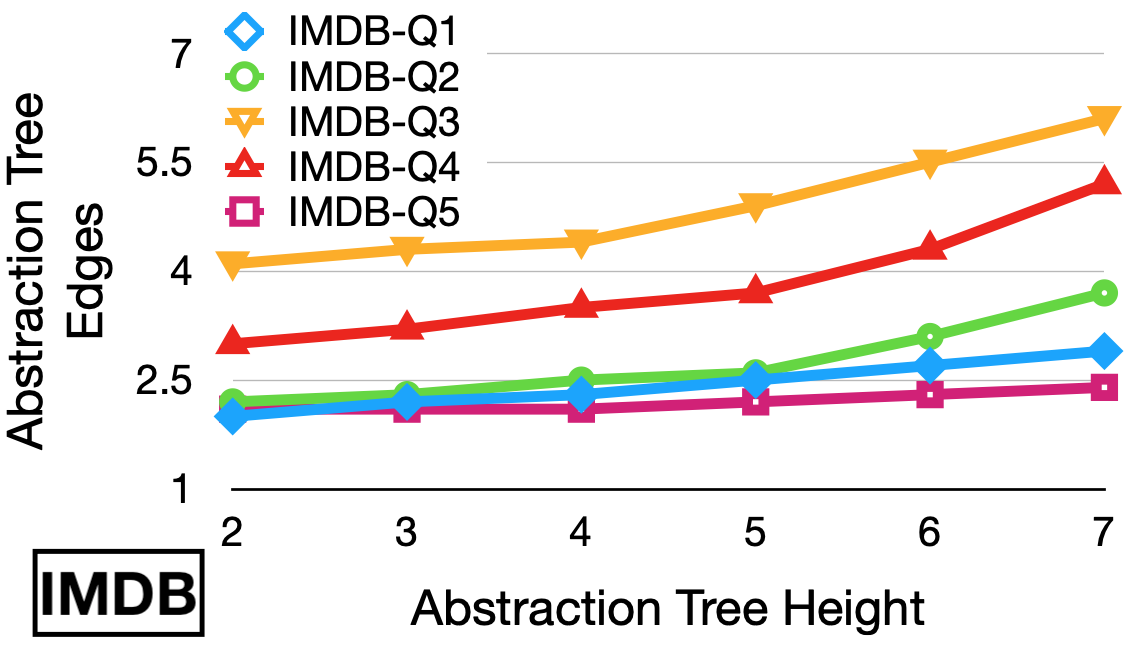}
	\end{subfigure}
	
	\caption{\reva{Optimal abstraction size for varying 
	tree height}}\label{fig:exp-tree-height-edges}
\end{figure}

{\bf Number of query joins (query complexity).~}
In this experiment we used TPC-H queries Q5, Q7, Q9, Q21 and \reva{IMDB queries Q2, Q4, Q7} (as this is the subset of queries with at least 6 joins) and examined the change in runtime as we increase the number of joins in each. We do so by starting with a version of the queries with only 3 joins and adding an atom for each tick on the X axis. 
The results (depicted in \reva{Figure \ref{fig:exp-joins-runtime}}) show that the runtime is not significantly affected by the increase in the number of joins. 

\begin{figure}[!htb]
	\centering
	\begin{subfigure}{0.49\linewidth}
        \centering
	    \includegraphics[width = \linewidth]{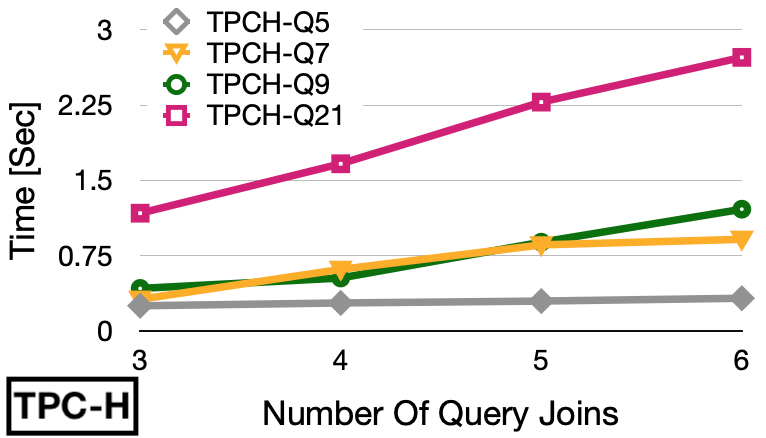}
	\end{subfigure}
	\hfill
	\begin{subfigure}{0.49\linewidth}
	    \centering
	    \includegraphics[width = \linewidth]{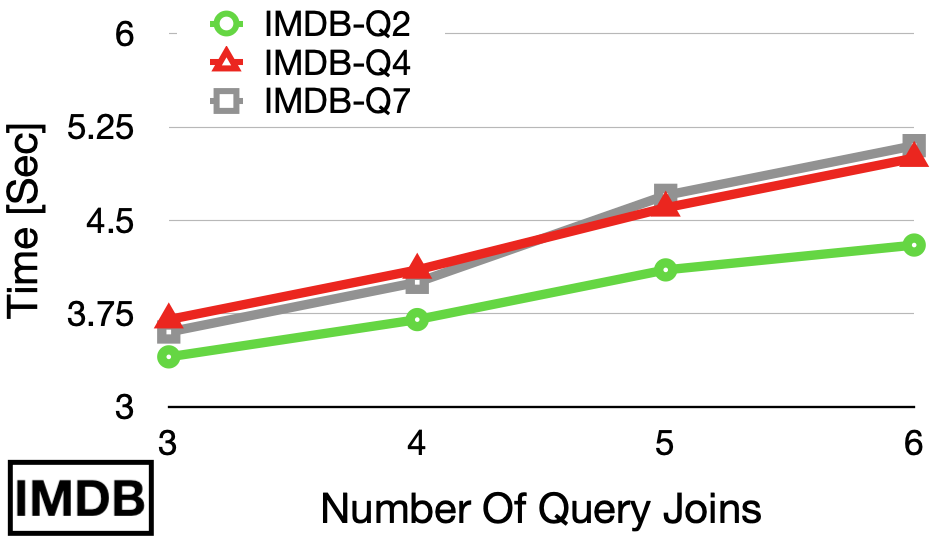}
	\end{subfigure}
	
	\caption{\reva{Runtime for varying number of joins}}\label{fig:exp-joins-runtime}
\end{figure}

{\bf \ex\ rows.~}
We examine our scalability in terms of increased the number of rows in the \ex. The results (shown in \reva{Figure \ref{fig:exp-ex-rows-runtime}}) indicate that the number of rows is a determining factor in the runtime of our algorithm. This is because a large number of rows implies fewer CIM queries for each concretization (since each row must be connected). Therefore, the algorithm was forced to try all possible (exponentially many) abstractions, similarly to the brute force approach, which significantly worsened the runtime.
In particular, for TPCH-Q21, the algorithm had to examine a large number of abstractions since this query includes three joined atoms with the `lineitem' relation, where each of them can be abstracted.

\begin{figure}[!htb]
	\centering
	\begin{subfigure}{0.49\linewidth}
        \centering
	    \includegraphics[width = \linewidth]{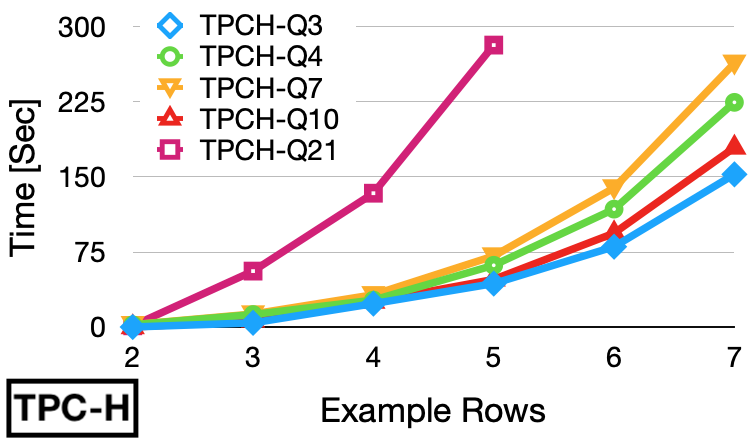}
	\end{subfigure}
	\hfill
	\begin{subfigure}{0.49\linewidth}
        \centering
	    \includegraphics[width = \linewidth]{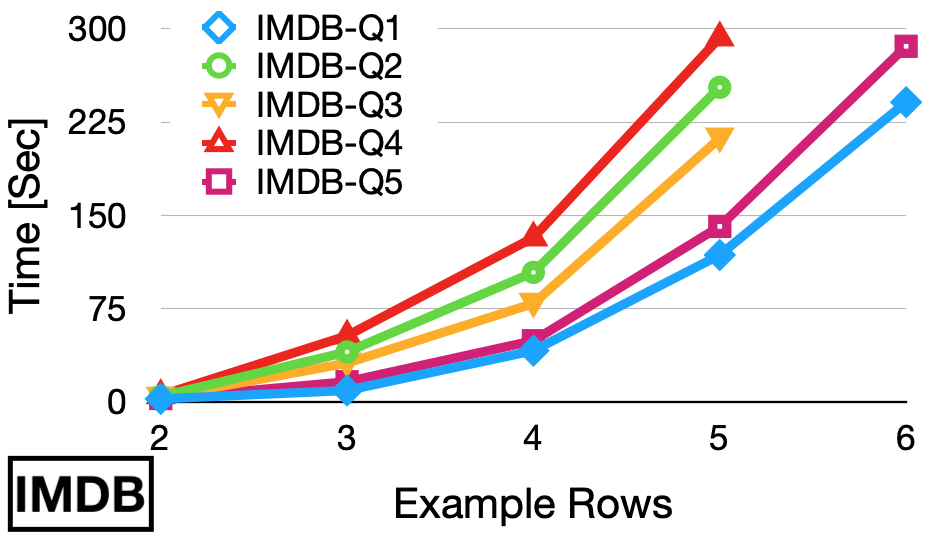}
	\end{subfigure}
	
	\caption{\reva{Runtime for varying \ex\ rows}}\label{fig:exp-ex-rows-runtime}
 	\vspace{-3mm}
\end{figure}



{\bf Loss of information distribution.~}
We have conducted all of the experiments for two loss of information distributions. The first is entropy with discrete uniform distribution and the second is entropy with random distribution (Section \ref{subsec:loss}). We found that on average, the runtimes are not affected by different distributions. As the probabilities change, the optimal abstraction for one distribution may not necessarily be the optimal for the other one. For example, if there is another abstraction with the same privacy, it may now have a smaller loss of information and will be the new optimal one.

\revb{{\bf Comparing to a different abstraction approach.~}
The notion of abstraction trees was presented in \cite{DeutchMR19}, where the goal of the abstraction was reducing the provenance size.
We used this approach to construct an alternative algorithm for our problem. Since the framework of \cite{DeutchMR19} was not designed to achieve privacy, we used it as a black-box, which we executed multiple times with a decreasing target provenance size, until we met the desired privacy threshold.
We compared the loss of information incurred by our algorithm to that of \cite{DeutchMR19}. The results are shown in Figure~\ref{fig:exp-compression}.
The compression-based approach of \cite{DeutchMR19} unnecessarily increases the loss of information by approximately $2\times$ to $3\times$ to achieve the same privacy as our approach.

\begin{figure}[!htb]
	\centering
	\begin{subfigure}{0.49\linewidth}
        \centering
	    \includegraphics[width = \linewidth]{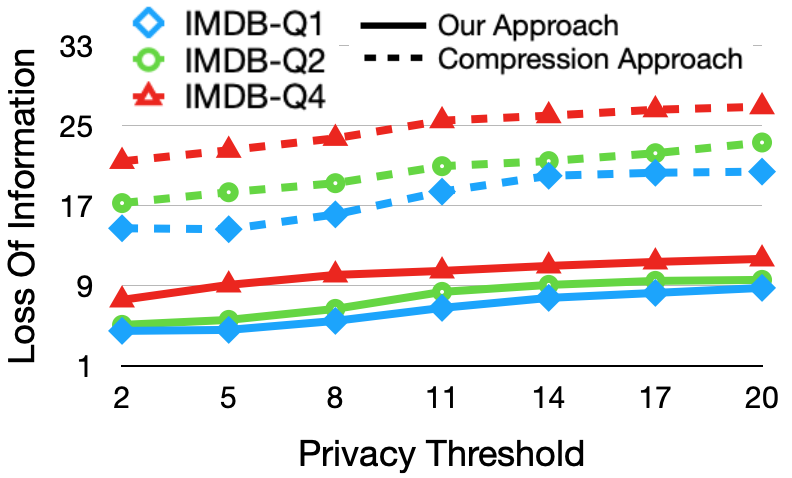}
	\end{subfigure}
	\caption{\revb{Loss of information for varying privacy thresholds, for our approach and the approach from \cite{DeutchMR19}}}\label{fig:exp-compression}
 	\vspace{-2mm}
\end{figure}

}

{\bf Effect of each algorithm component.~}
We now present the effects on the execution time of the five algorithm components we have detailed in Section \ref{sec:algo-high-level}, compared to a brute-force approach. 
\reva{The effect of each component is measured as a standalone optimization.}
Figure \ref{fig:improvements} shows the results for each component.
Referring to the names of the components in Section \ref{sec:algo-high-level}, `Sorting the abstractions' and `Prioritizing loss of information over privacy computation' have improved performance by a factor of over $500\times$. The third component of `Computing privacy row by row' has improved performance by approximately $2\times$ to $4\times$ for a \ex\ with three rows. For a \ex\ with four rows, it improved performance by approximately $10\times$ to $100\times$. For \ex\ with more than five rows we were unable to find a solution to the problem in a reasonable time using the brute force approach, in contrast to our approach. The fourth component, `Concretizations connectivity', has improved performance by approximately $1.5\times$ to $1.8\times$ when we filtered out about $60\%$ of the concretizations. The last component, `Caching information about concretizations and queries', has improved performance by approximately $1.5\times$~to~$4\times$.

\begin{figure}[!htb]
    \centering 
	\includegraphics[width=0.7\linewidth] {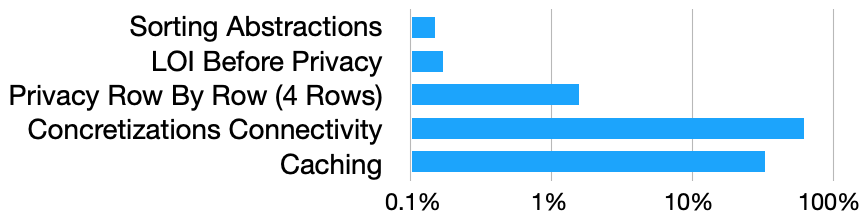}
	\caption{Effect of each of algorithm component from Section \ref{sec:algo-high-level} as compared to the brute force approach (brute force execution time is marked by 100\%)}
    \label{fig:improvements}
    \vspace{-4mm}
\end{figure}

\common{

\begin{table}[h]
    \centering \footnotesize
    \caption{\common{User Study Results Summary}}\label{tbl:user-study-results}
    \begin{tabularx}{\linewidth}{| X | c | c | c | c | c |}
        \cline{2-3}\multicolumn{1}{c|}{} & 
        \begin{tabular}{@{}c@{}}{\bf Group A}\end{tabular} & 
        \begin{tabular}{@{}c@{}}{\bf Group B}\end{tabular} \\
        \hline Number of group members that were able to find the original query & 6/6 (100\%) & 0/6 (0\%) \\
        \hline Number of correct answers in hypothetical questions (on average) & 9.6/10 (96\%) & 8.5/10 (85\%) \\
        \hline
    \end{tabularx}
\end{table}

{\bf User Study:~} \label{sec:user-study}
We have conducted a user study,
involving 12 users with knowledge of databases. 
The users were randomly divided into two groups of equal size: control group (Group A) and treatment group (Group B). 
We used IMDB-Q3 (all the actors who played in a movie with the actor Kevin Bacon), the IMDB abstraction tree, 2 rows of output, and a privacy threshold of 2. Then, with Algorithm \ref{algo:optimal-abs} we found the optimal abstraction.
Group A was given the output with the original provenance while group B received the output with the abstracted provenance and the abstraction tree. 
The users were given two tasks: (1) Infer the underlying query from the original (Group A)/abstracted (Group B) provenance and (2) Answer $10$ hypothetical questions regarding the effect of deleting rows (e.g., regarding action movies) from the database on the query result. 
The study results are summarized in Table~\ref{tbl:user-study-results}. 

For the first task, all members of group A and none of the members of group B were able to identify the original query.
For the second task, the members of group A were able to answer on average 9.6 out of 10 questions correctly, while the members of group B were able to answer on average 8.5 out of 10 questions. This shows a reasonable loss of information. The breakdown of correct answers is shown in Figure \ref{fig:user-study-questions} and indicates the following conclusions.
In most cases, the abstracted provenance has provided enough information to answer the question. For example, for question Q6, which considers the effect of the removal of all comedy movies released after 1980, the abstracted provenance could be used to determine the correct answer. This is because the abstracted value that replaced the relevant tuple was ``comedy movie released in 1990--2000''.
In some cases, there were a few mistakes due to misunderstandings or lack of concentration.
In contrast, naturally, there were cases where the abstracted provenance was not detailed enough to answer the question. For instance, question Q9, that refers to a case where directors born before 1970 are removed from the database. The abstracted provenance indicated that the output is related to a person born between 1950 to 1960, but not to the person's role in the movie (actor or director), thus the members of group B were unable to answer the question.
Overall, our user study indicates that our method was successful in hiding the original query and incurred a reasonable loss of information in terms of using provenance.


\begin{figure}[!htb]
    \centering
	\includegraphics[width=0.65\linewidth] {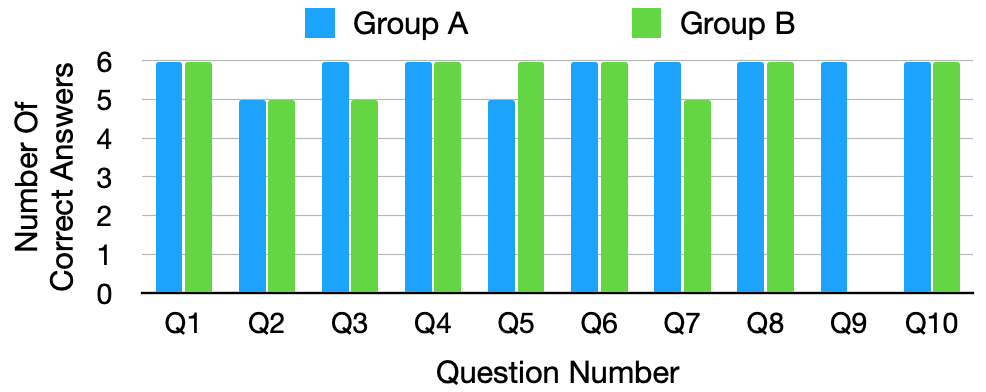}
	\caption{\common{Breakdown of correct answers in hypothetical questions of the user study}}
    \label{fig:user-study-questions}
\end{figure}
\vspace{-10mm}
}

\section{related work}\label{sec:related}
We next review previous work in the fields of
provenance and privacy, highlighting our novelty. 


There is a wealth of works on data provenance and its uses, including
relational algebra, XML query languages, Nested
Relational Calculus, and functional programs (see e.g.,
\cite{trio,GKT-pods07,Userssemiring1,CheneyProvenance,ProvenanceBuneman,Olteanu12,GS13,Tan03} and a survey \cite{HerschelDL17}). 
These works have generally focused on provenance modeling, efficient tracking and storage, and algorithms that use provenance for different applications. As such, they are orthogonal to our work: extending our solutions to additional query and provenance formalisms proposed in these works is an important challenge for future work.

The area of privacy and security in the context of provenance has been explored by various works \cite{DavidsonKMPR11,DavidsonKRSTC11,DavidsonKTRCMS11,BertinoGKNPSSTX14,SanchezBS18,Ruan0DLOZ19,AndersonC12,TanKH13,Cheney11}. 
These works have focused on privacy and security in different settings than ours such as IoT \cite{SanchezBS18}, Blockchain \cite{Ruan0DLOZ19} and workflows \cite{DavidsonKMPR11,DavidsonKRSTC11,DavidsonKTRCMS11}, while our focus was the relational setting. The difference in the setting is reflected in the provenance models (we focus on provenance polynomials whereas, e.g., \cite{DavidsonKMPR11} focuses on workflow provenance in the form of input-output relationship between modules). In turn, the technical problems and solutions are inherently different.

A recent work on fine-grained provenance privacy \cite{DeutchG19} has focused on learning queries from \ex s where the provenance is given in different semirings \cite{GKT-pods07,Greenicdt09}. 
It showed that reducing the granularity of the provenance by using less detailed semirings (which may be seen as an alternative to our approach of abstracting provenance expressions) is inadequate for privacy purposes: it does not introduce significant added difficulty when attempting to reverse-engineer the underlying query.

In \cite{DavidsonKMPR11,DavidsonKRSTC11,DavidsonKTRCMS11} the authors studied workflow privacy, with a privacy criterion inspired by $l$-diversity \cite{MachanavajjhalaGKV06} and $k$-anonymity~\cite{KAnonymity2002}. 
This model achieves privacy by obfuscating entire attributes of a relation that represents a workflow.
In contrast, we do not focus on black-box modules, but rather on detailed fine-grained provenance obtained from queries. This makes the technical results of these works inapplicable to our setting. 
The work of~\cite{Cheney11} has described an abstract framework for provenance security and defines the notions of the disclosure and obfuscation properties of provenance. Given a query and two traces, the problem is then to determine whether the output of the query is equal on these two traces, if they have the same provenance view. 
A prominent difference from our model is the assumption that the underlying query is known which makes the problem definition and solution fundamentally different. 



Previous work on abstracting provenance has primarily focused on workflow provenance abstractions and graph abstractions \cite{CheneyP14,CadenheadKKT11,DeyZL11,ChebotkoCLFY08,BitonBDH08,ChebotkoLCFY10,DeutchMR19}, mainly for the purpose of reducing the provenance size and/or optimizing its generation.  
Security Views \cite{ChebotkoCLFY08} is a framework for access control where users can specify the desired security of the components of a scientific workflow. The framework then omits the inaccessible components from the provenance view. ZOOM \cite{BitonBDH08} abstracts the provenance view by grouping models together allowing users to focus only on the relevant part of the workflow, and ProPub \cite{DeyZL11} allows users to publish provenance while anonymizing, abstracting, or hiding parts of the provenance graph. Here again, the  models (coarse-grained workflow provenance models) and problems that are studied in these works significantly differ from those of the present work.

\revb{
Query reverse-engineering from output examples \cite{Shen,KalashnikovLS18,joinQueries,qbo,TanZES17} attempts to assist users who lost access to the original query or want an automatic system to infer a query based on output examples. In the context of our work, such systems may be of use in the computation of privacy when the provenance is given in the \lin\ semiring, as mentioned in Section \ref{sec:algo}. This is an intriguing subject of future work.
}


\section{Conclusion and Limitations}
\label{sec:conc}

We have proposed in this paper a novel framework for striking a balance between utility and privacy when releasing data provenance. The framework is based on obfuscating provenance by identifying annotations appearing in it, thereby hiding to some extent the query whose execution has yielded the provenance. This kind of obfuscation may be done in many ways, and we aim at choosing the optimal one. The resulting problem is NP-hard, yet we have provided practically effective heuristics.

There are many important directions for future work. 
First, our work assumes an abstraction tree as input, which may not be readily given. (Semi-)automatic inference of abstraction trees, as briefly discussed in the paper, is an important complementary problem. 
Second, our loss-of-information model relies on a probability distribution over the leaves, and in our experiments, we have mostly assumed a uniform distribution; we intend to study means for inferring probabilities, as well as other weight-based models for loss of information.  Third, our model is tailored to the provenance semiring model; studying provenance obfuscation in the context of other provenance models is another intriguing goal for future research. 

\clearpage

\bibliographystyle{abbrv}
\bibliography{bibtex.bib}

\end{document}